\documentclass{article}

\usepackage{arxiv}

\usepackage[numbers]{natbib}

\usepackage[utf8]{inputenc}
\usepackage[T1]{fontenc}
\usepackage{microtype}
\usepackage{siunitx}
\newcommand{\code}[1]{\texttt{\small #1}}

\usepackage{array}                     
\usepackage{booktabs}                  
\usepackage{tabularx}
\usepackage{longtable}
\usepackage{multirow}
\usepackage{makecell}
\usepackage{rotating}
\usepackage{colortbl}                  

\newcolumntype{L}[1]{>{\raggedright\arraybackslash}p{#1}}
\newcolumntype{C}[1]{>{\centering\arraybackslash}p{#1}}

\usepackage{enumitem}
\usepackage{siunitx}
\usepackage{comment}
\usepackage{setspace}
\usepackage{url}                       

\usepackage{amsmath}                   
\usepackage{amssymb}                   
\usepackage{amsfonts}                  
\usepackage{amsthm}                    
\usepackage{nicefrac}                  

\newtheorem{theorem}{Theorem}
\newtheorem{lemma}{Lemma}
\theoremstyle{definition}
\newtheorem{definition}{Definition}

\usepackage{graphicx}
\usepackage{tikz}
\usetikzlibrary{
    positioning,
    shapes.geometric,
    shapes.misc,
    arrows.meta,
    calc,
    mindmap,
    shadows,
    fit,
    backgrounds,
    patterns
}
\usepackage{pgfplots}
\pgfplotsset{compat=1.18}

\usepackage{algorithm}
\usepackage{algpseudocode}

\definecolor{unsafrd} {HTML}{B71C1C}
\definecolor{figteal}  {RGB}{180,225,210}
\definecolor{codebg}    {HTML}{F8F8F8}
\definecolor{codeframe} {HTML}{DDDDDD}
\definecolor{codegreen} {HTML}{2E7D32}
\definecolor{codeblue}  {HTML}{1565C0}
\definecolor{codegray}  {HTML}{757575}
\definecolor{codepurple}{HTML}{7B1FA2}

\usepackage{listings}
\usepackage{hyperref}

\usepackage[acronym, nogroupskip, nonumberlist, nopostdot]{glossaries}
\glsdisablehyper   

\newcommand{\arduinotool}{ESBMC-Arduino}

\providecommand{\Description}[1]{}

\title{ESBMC-Arduino: Closing the Deployment Gap for Formal Verification of Open-Hardware PLCs}

\renewcommand{\shorttitle}{ESBMC-Arduino: Closing the Deployment Gap for Open-Hardware PLCs}
\renewcommand{\undertitle}{Preprint}

\author{%
  \href{https://orcid.org/0000-0001-6390-9340}{Pierre Dantas}%
    \thanks{Corresponding author.} \\
  Department of Computer Science \\
  The University of Manchester \\
  Manchester, UK \\
  \texttt{pierre.dantas@manchester.ac.uk} \\
  \And
  \href{https://orcid.org/0000-0002-6235-4272}{Lucas Cordeiro} \\
  Department of Computer Science \\
  The University of Manchester \\
  Manchester, UK \\
  \texttt{lucas.cordeiro@manchester.ac.uk} \\
  \And
  \href{https://orcid.org/0000-0003-3095-0042}{Waldir Junior} \\
  Electrical Engineering \\
  Federal University of Amazonas (UFAM) \\
  Manaus, AM, Brazil \\
  \texttt{waldirjr@ufam.edu.br} \\
}

\hypersetup{
  pdftitle  = {ESBMC-Arduino: Closing the Deployment Gap for Formal Verification of Open-Hardware PLCs},
  pdfauthor = {Pierre Dantas, Lucas Cordeiro, Waldir Junior},
  pdfkeywords = {deployment gap, hardware-faithful verification, HAL descriptor, machine word width, ADC input model, PLC, IEC 61131-3, ESBMC, BMC, k-induction, formal verification, PLCopen XML},
}

\begin{document}

\newacronym{aadl}{AADL}{Architecture Analysis and Design Language}
\newacronym{abi}{ABI}{Application Binary Interface} 
\newacronym{ansi-c}{ANSI-C}{American National Standards Institute C}
\newacronym{api}{API}{Application Programming Interface}
\newacronym{adc}{ADC}{Analog-to-Digital Converter}
\newacronym{ast}{AST}{Abstract Syntax Tree}
\newacronym{bdd}{BDD}{Binary Decision Diagrams}
\newacronym{bmc}{BMC}{Bounded Model Checking}
\newacronym{cbmc}{CBMC}{Bounded Model Checking for ANSI-C Programs}
\newacronym{cegar}{CEGAR}{Counterexample-Guided Abstraction Refinement}
\newacronym{cern}{CERN}{Conseil Européen pour la Recherche Nucléaire}
\newacronym{cfg}{CFG}{Control Flow Graph}
\newacronym{chc}{CHC}{Constrained Horn Clause}
\newacronym{cli}{CLI}{Command-Line Interface}
\newacronym{cpu}{CPU}{Central Processing Unit}
\newacronym{ctl}{CTL}{Computation Tree Logic}
\newacronym{cuda}{CUDA}{Compute Unified Device Architecture}
\newacronym{cve}{CVE}{Common Vulnerability and Exposure}
\newacronym{dfs}{DFS}{Depth-First Search}
\newacronym{dsl}{DSL}{Domain-Specific Language}
\newacronym{epsrc}{EPSRC}{Engineering and Physical Sciences Research Council}
\newacronym{esbmc}{ESBMC}{Efficient SMT-based Context-Bounded Model Checker}
\newacronym{fbd}{FBD}{Functional Block Diagram}
\newacronym{fpga}{FPGA}{Field-Programmable Gate Array}
\newacronym{gpio}{GPIO}{General-Purpose Input/Output}
\newacronym{hal}{HAL}{Hardware Abstraction Layer}
\newacronym{hil}{HIL}{Hardware-in-the-Loop}
\newacronym{ic3}{IC3}{Incremental Construction of Inductive Clauses for Indubitable Correctness}
\newacronym{ide}{IDE}{Integrated Development Environment}
\newacronym{iec}{IEC}{International Electrotechnical Commission}
\newacronym{ieee}{IEEE}{Institute of Electrical and Electronics Engineers}
\newacronym{ics}{ICS}{Industrial Control Systems}
\newacronym{il}{IL}{Instruction List}
\newacronym{iot}{IoT}{Internet of Things}
\newacronym{ir}{IR}{Intermediate Representation}
\newacronym{iso}{ISO}{International Organization for Standardization}
\newacronym{ld}{LD}{Ladder Diagram}
\newacronym{llb}{LLB}{Ladder Logic Bombs}
\newacronym{llm}{LLM}{Large Language Model}
\newacronym{ltl}{LTL}{Linear Temporal Logic}
\newacronym{mcu}{MCU}{Microcontroller Unit}
\newacronym{nasa}{NASA}{National Aeronautics and Space Administration}
\newacronym{pwm}{PWM}{Pulse-Width Modulation}
\newacronym{pdr}{PDR}{Property Directed Reachability}
\newacronym{pid}{PID}{Proportional-Integral-Derivative}
\newacronym{plc}{PLC}{Programmable Logic Controller}
\newacronym{pou}{POU}{Program Organization Unit}
\newacronym{por}{POR}{Partial Order Reduction}
\newacronym{rtos}{RTOS}{Real-Time Operating System}
\newacronym{sat}{SAT}{Boolean Satisfiability}
\newacronym{scl}{SCL}{Structured Control Language}
\newacronym{sfc}{SFC}{Sequential Function Chart}
\newacronym{slr}{SLR}{Systematic Literature Review}
\newacronym{smt}{SMT}{Satisfiability Modulo Theories}
\newacronym{smtlib2}{SMT-LIB}{Satisfiability Modulo Theories Library}
\newacronym{ssa}{SSA}{Static Single Assignment}
\newacronym{st}{ST}{Structured Text}
\newacronym{stl}{STL}{Statement List}
\newacronym{svcomp}{SV-COMP}{Competition on Software Verification}
\newacronym{tacas}{TACAS}{Tools and Algorithms for the Construction and Analysis of Systems}
\newacronym{ufam}{UFAM}{Federal University of Amazonas}
\newacronym{ukri}{UKRI}{UK Research and Innovation}

\maketitle

\begin{abstract}
OpenPLC, Arduino OPTA, CONTROLLINO, and Industrial Shields M-Duino -- bring \acrshort{iec}~61131-3 programming to low-cost, microcontroller-class devices increasingly used in real automation and in \gls{ics} security research. Existing open-source verifiers for \acrshort{iec}~61131-3, including the ESBMC-PLC line, prove safety over an \emph{abstract} scan-cycle model with idealized unbounded integers. The artifact that runs on the board, however, executes on a resource-constrained \gls{mcu} with a 16-bit machine word (8-bit AVR Arduinos), and sensors are read through a finite-resolution \gls{adc}. We show that this \emph{deployment gap} makes naive width-aware verification \emph{unsound in practice}: on 123 real third-party programs, checking 16-bit arithmetic overflow without a hardware input model raises a \SI{44}{\percent} false-alarm rate (54/123) and finds zero genuine defects, because it explores sensor values no \gls{adc} can produce. Because the gap sits exactly where computation meets the physical process -- a bounded sensor reading scaled by finite-width integer arithmetic into an actuation command -- an overflow can silently suppress a physical safety action such as a high-level alarm. At the same time, an unbounded input model fabricates alarms that no environment can trigger. We present a method for \textbf{hardware-faithful verification of \acrshort{iec}~61131-3 on open hardware}: a declarative \gls{hal} descriptor (word width, \gls{adc}/\gls{pwm} resolution, I/O binding) and a sound lowering that interprets arithmetic at the target word width and constrains each input to its hardware-realizable range, transferring to the \gls{plc} domain the target-model discipline established for safety-critical embedded C (Frama-C, Astr\'ee). We instantiate it for Arduino as \textbf{\arduinotool{}}, deriving \gls{hal} parameters from the official Arduino cores and realizing the input-range model as a lightweight hook in the \gls{esbmc} \gls{ld} frontend. On the 123-program corpus, the \gls{hal} annotator \textbf{eliminates all 54 false alarms (54$\to$0)} while preserving every robustness proof, and a controlled corpus demonstrates the genuine -- but rare -- width-dependent defects the method detects with physically realizable witnesses.
\end{abstract}


\glsresetall

\section{Introduction}
\label{sec:intro}

\glspl{plc} govern safety-critical processes, but a new class of \emph{open-hardware} \glspl{plc} -- OpenPLC~\cite{Alves2018openplc, Alves2014openplc}, Arduino OPTA~\cite{ArduinoOpta, ArduinoPLCIDE}, CONTROLLINO~\cite{Controllino}, and Industrial Shields M-Duino~\cite{IndustrialShields} -- now brings \gls{iec}~61131-3 programming to low-cost, microcontroller-class devices. These platforms are widely used in education, small-scale automation, and -- importantly for assurance -- as the runtime substrate for academic \gls{ics} security testbeds; OpenPLC, in particular, was designed explicitly for cybersecurity research~\cite{Alves2018openplc}.

\subsection{The Deployment Gap}
Open-source verifiers for \gls{iec}~61131-3 -- the ESBMC-PLC in particular~\cite{DantasCordeiro2026artefact, DantasCordeiro2026graphical, ESBMCpr5400}, and PLCverif~\cite{LopezMiguel2022,LopezMiguel2025} -- prove safety properties over an \emph{abstract scan-cycle model} of the \gls{ld} or \gls{st} logic, with idealized, effectively unbounded integers and unconstrained inputs. The artifact that runs on an open-hardware \gls{plc} is different: it executes on a resource-constrained \gls{mcu} where a machine word is only 16~bits on the 8-bit AVR boards (Arduino Uno/Nano/Mega), and physical inputs arrive through an \gls{adc} of finite resolution (a 10-bit \gls{adc} yields integer readings in $[0,\num{1023}]$). Arithmetic that is exact in the abstract model can therefore \emph{overflow} on the device, and a property proved on the logic model need not hold on the deployed binary. We call this discrepancy the \emph{deployment gap}.

The gap cuts both ways, and both directions matter for assurance. A program can be unsafe on the device yet pass abstract verification (a missed defect); and -- less obviously but, as we show, far more common -- \emph{naively} accounting for the word width without modelling the hardware \emph{inputs} makes verification \emph{unsound in practice}: the verifier explores sensor values no \gls{adc} can produce (e.g.\ IN = \num{-32764}), reports overflows that can never occur, and buries any real finding under false alarms. On 123 real third-party programs, we measured a \SI{44}{\percent} false-alarm rate (54/123) with zero genuine defects found this way (\S\ref{sec:eval}).

\subsection{Contributions}
This paper closes the gap with a method and its open-hardware instantiation. Our thesis is \emph{``what you verify must be what runs on the device''}, which requires modeling \emph{both} the target word width and the hardware-realizable input ranges.

\begin{enumerate}[leftmargin=2em]
    \item \textbf{Hardware-faithful verification of \gls{iec}~61131-3 (\S\ref{sec:method}).} A method that, given a program $P$ and a declarative \gls{hal} descriptor $H$ (word width, \gls{adc}/\gls{pwm} resolution, I/O binding), constructs a sound \emph{deployment model} $M_{dev}(P,H)$: arithmetic is interpreted at the target word width and every input is constrained to its hardware-realizable range. We give the construction, a soundness argument, and a gap-diagnosis procedure that flags hardware-induced defects with \emph{physically realizable} witnesses. Methodologically, this transfers to the \gls{plc} domain the target-model discipline long established for safety-critical embedded~C (Frama-C's machine model~\cite{Kirchner2015framac}, Astr\'ee~\cite{Cousot2005astree}), with the novelty that the \gls{hal} model is derived \emph{automatically} from the program's PLCopen I/O addresses.
    
    \item \textbf{\arduinotool{}: an Arduino instantiation (\S\ref{sec:arduino}).} A \gls{hal} library whose parameters are sourced from the official Arduino cores~\cite{ArduinoCoreAVR}, and a realization of the input-range model as a lightweight ($\approx$40-line) hook in the \gls{esbmc}~\gls{ld} frontend -- the \gls{smt} backend and \textit{k}-induction engine are reused unchanged.
    
    \item \textbf{Evaluation on real and controlled corpora (\S\ref{sec:eval}).} On 123 real third-party \gls{iec}~61131-3 programs, naive 16-bit overflow checking yields a \SI{44}{\percent} false-alarm rate and zero true findings; the \gls{hal} annotator \textbf{eliminates all 54 false alarms (54$\to$0)} while preserving every robustness proof. A controlled corpus demonstrates the genuine width-dependent defects the method detects (with realizable witnesses) and precisely characterizes which arithmetic patterns are vulnerable, thereby establishing, honestly, that the gap is real but \emph{rare} in the public corpora we could ingest.
\end{enumerate}

\subsection{Scope}
We target the integer/Boolean fragment of \gls{iec}~61131-3, deployed via open-hardware toolchains, and evaluate at the model level (no hardware-in-the-loop). Non-linear floating-point function blocks (e.g.\ \gls{pid}, transcendental sensor linearization) are outside the deployment-gap claim and are soundly over-approximated (\S\ref{sec:method}); we discuss this and other limitations in \S\ref{sec:threats}.

\subsubsection*{Why it Matters} 
Open-hardware \glspl{plc} are not toys: they run real water, lighting, and process-control logic, and they are the default substrate for the academic \gls{ics} security testbeds on which much of the field's empirical work depends. A safety or security argument built on abstract verification of such a system is only as sound as the assumption that the abstract model matches the device -- an assumption this paper shows to be false in a precise, reproducible way, and then repairs. The repair is cheap (a four-field descriptor and a 40-line hook) and reuses a mature verifier, so it is realistic to adopt rather than a clean-slate proposal.

\subsubsection*{Roadmap}
\S\ref{sec:background} reviews the scan cycle, the OpenPLC toolchain, the machine models of Arduino targets, and the verification backend. \S\ref{sec:related} positions the work against \gls{plc} verification and target-faithful embedded-C analysis. \S\ref{sec:gap} formalizes the deployment gap; \S\ref{sec:method} presents the method (descriptor, lowering, soundness; Fig.~\ref{fig:arch}); \S\ref{sec:arduino} gives the Arduino instantiation and implementation; \S\ref{sec:walkthrough} walks through two end-to-end scenarios; and \S\ref{sec:eval} reports the evaluation. \S\ref{sec:discussion}--\ref{sec:future} discuss implications, threats, and future work.

\section{Background}
\label{sec:background}

\subsection{The PLC Scan Cycle}
A \gls{plc} executes a deterministic \emph{scan cycle}: it reads all inputs into a process image, executes the user program once from first rung to last, writes the output image to the actuators, and repeats~\cite{ESBMCpr5400}. The ESBMC-PLC line encodes this as a \code{while(true)} loop in which inputs are re-sampled each iteration (an open-world sensor model) nondeterministically, output coils are persistent state, and safety properties are injected as assertions; \textit{k}-induction then proves invariants for all future scans. We reuse this encoding and add a faithful model of the hardware boundary that the scan cycle crosses on a real device.

\subsection{Open-Hardware PLCs and the OpenPLC Toolchain}
Open-hardware \glspl{plc} run \gls{iec}~61131-3 programs on commodity microcontroller boards. \textbf{OpenPLC}~\cite{Alves2018openplc, Alves2014openplc} is the dominant open-source example: its editor exports PLCopen XML, and its runtime compiles the program (via MATIEC~\cite{deSousa2014}) to C, which is then cross-compiled for Arduino-, ESP32-, or Raspberry Pi-class targets. \textbf{Arduino OPTA}~\cite{ArduinoOpta} is Arduino's industrial micro-\gls{plc}, programmed in \gls{iec}~61131-3 via the Arduino PLC IDE~\cite{ArduinoPLCIDE}; \textbf{CONTROLLINO}~\cite{Controllino} and \textbf{Industrial Shields M-Duino}~\cite{IndustrialShields} are Arduino-compatible industrial \glspl{plc} commonly driven by OpenPLC. All reduce a vendor-neutral program to C plus a hardware binding (the object we verify).

\subsection{The \acrfull{hal}: Word Width, Pin Binding, ADC, and PWM}
The \gls{hal} connects logical \gls{iec}~61131-3 variables to physical peripherals. Each Boolean \code{\%IX}/\code{\%QX} address binds to a \gls{gpio} pin; analog inputs are read through an \gls{adc} of finite resolution (typically 10--14 bits) and scaled; analog outputs are produced by \gls{pwm}. Two facts of this layer drive the deployment gap: (i)~the \gls{mcu}'s native word width -- 16-bit \code{int} on AVR-class boards -- so computations exact in the abstract model can overflow; and (ii)~the bounded \emph{input domains} -- an $n$-bit \gls{adc} reading lies in $[0,2^{n}-1]$, a digital input in $\{0,1\}$ -- which the abstract model omits entirely (Fig.~\ref{fig:hal}).

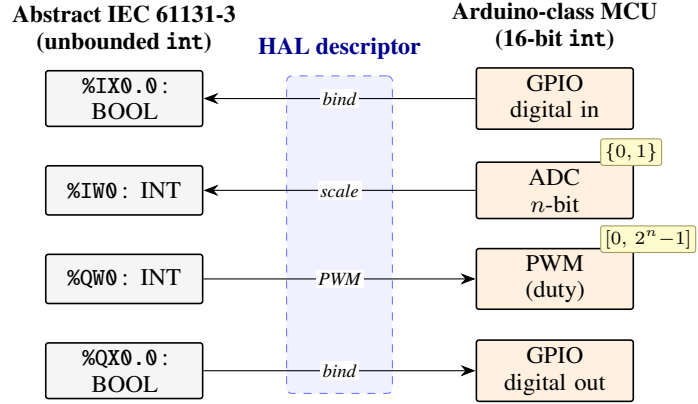
\begin{figure}[htbp]
    \centering
    \begin{tikzpicture}[
          font=\small, >={Stealth[length=2mm]},
          log/.style={draw, semithick, rounded corners=1pt, fill=gray!8,
                      align=center, inner sep=2.5pt, text width=19mm, minimum height=6.5mm},
          hw/.style ={draw, semithick, rounded corners=1pt, fill=orange!12,
                      align=center, inner sep=2.5pt, text width=19mm, minimum height=6.5mm},
          bind/.style={font=\scriptsize\itshape, fill=white, inner sep=1pt},
          dom/.style={font=\scriptsize, fill=yellow!25, draw=yellow!55!black,
                      rounded corners=1pt, inner sep=1.5pt},
          hdr/.style={font=\footnotesize\bfseries, align=center},
        ]
        \def\dx{36mm}
        \node[log] (ix) {\code{\%IX0.0}\,: BOOL};
        \node[log, below=5mm of ix] (iw) {\code{\%IW0}\,: INT};
        \node[log, below=5mm of iw] (qw) {\code{\%QW0}\,: INT};
        \node[log, below=5mm of qw] (qx) {\code{\%QX0.0}\,: BOOL};
        \node[hw, right=\dx of ix] (gpi) {GPIO\\digital in};
        \node[hw, right=\dx of iw] (adc) {ADC\\$n$-bit};
        \node[hw, right=\dx of qw] (pwm) {PWM\\(duty)};
        \node[hw, right=\dx of qx] (gpo) {GPIO\\digital out};
        \coordinate (mt) at ($(ix)!0.5!(gpi)$);
        \coordinate (mb) at ($(qx)!0.5!(gpo)$);
        \begin{scope}[on background layer]
          \node[draw=blue!55, dashed, fill=blue!6, rounded corners=2pt,
                inner xsep=7mm, inner ysep=3mm, fit=(mt)(mb)] (hal) {};
        \end{scope}
        \node[hdr, blue!55!black, above=1mm of hal.north] {\gls{hal} descriptor};
        \node[hdr, above=1mm of ix] {Abstract \acrshort{iec}~61131-3\\(unbounded \code{int})};
        \node[hdr, above=1mm of gpi] {Arduino-class \acrshort{mcu}\\(16-bit \code{int})};
        \draw[<-] (ix) -- node[bind]{bind} (gpi);
        \draw[<-] (iw) -- node[bind]{scale} (adc);
        \draw[->] (qw) -- node[bind]{PWM} (pwm);
        \draw[->] (qx) -- node[bind]{bind} (gpo);
        \node[dom, below right=1.1mm and 6mm of gpi.south] {$\{0,1\}$};
        \node[dom, below right=0.8mm and 6mm of adc.south] {$[0,\,2^{n}{-}1]$};
    \end{tikzpicture}
    \caption{The \gls{hal} binds logical \gls{iec}~61131-3 variables to physical peripherals: Boolean \code{\%IX}/\code{\%QX} addresses map to \gls{gpio} pins, analog inputs are read through a finite-resolution \gls{adc} and scaled, analog outputs are produced by \gls{pwm}. Two facts about this layer drive the deployment gap: the \gls{mcu}'s native 16-bit word width and the bounded input domains that the abstract model omits.}
    \Description{Diagram comparing abstract IEC 61131-3 logical variables on the left with physical Arduino-class MCU peripherals on the right. The left column, titled 'Abstract IEC 61131-3 (unbounded int)', contains four stacked boxes: percent IX0.0 (BOOL), percent IW0 (INT), percent QW0 (INT), and percent QX0.0 (BOOL). The right column, titled 'Arduino-class MCU (16-bit int)', contains four corresponding stacked boxes: GPIO digital in, ADC n-bit, PWM (duty), and GPIO digital out. A blue dashed rectangle labeled '\gls{hal} descriptor' spans the gap between the two columns behind the arrows. Arrows connect each left variable to its right peripheral: input variables (percent IX0.0 and percent IW0) have arrows pointing from the MCU toward the logic, labeled 'bind' for the GPIO and 'scale' for the ADC; output variables (percent QW0 and percent QX0.0) have arrows pointing from the logic toward the MCU, labeled 'PWM' for the duty cycle and 'bind' for the GPIO. Under the GPIO digital in box, a yellow box labeled '0,1' appears. Under the ADC box, a yellow box states '0, 2^n-1'. At the bottom, a red-highlighted box titled 'Deployment gap' explains two facts: (i) the MCU uses 16-bit int, so abstract computations can overflow; and (ii) inputs are bounded to finite ranges, which the abstract model omits.
    }
    \label{fig:hal}
\end{figure}

\subsection{PLCopen XML}
\gls{ld} programs are exchanged as PLCopen XML in textual and graphical encodings, reconstructed into rungs by the ESBMC-GraphPLC \gls{dfs} resolver~\cite{DantasCordeiro2026graphical}. Open-hardware editors export the graphical encoding; crucially, the variable declarations carry explicit \code{\%IX}/\code{\%IW} addresses, which we exploit as the ground truth for \gls{hal} input binding. The \gls{iec}~61131-3 direct-addressing scheme is itself the bridge between logic and hardware: a location prefix names the memory area (\code{I}~input, \code{Q}~output, \code{M}~internal) and a size prefix names the width (\code{X}~bit, \code{B}~byte, \code{W}~word, \code{D}~double word). Thus \code{\%IX0.0} is a single digital input bit, while \code{\%IW0} is a 16-bit analog input word; on an open-hardware target, the former is wired to a \gls{gpio} pin read by \code{digitalRead} and the latter to an \gls{adc} channel read by \code{analogRead}. The address, therefore, encodes the information our method needs precisely -- whether a variable is an input and, if analog, the width of the channel feeding it -- without any additional annotation by the Engineer.

\subsection{From PLCopen XML to a Binary: the OpenPLC Toolchain}
\label{sub:openplc-pipeline}
Understanding where the deployment gap enters requires following a program from the editor to the device. In the dominant open-source flow, the OpenPLC Editor exports a PLCopen XML project; the OpenPLC runtime invokes the MATIEC compiler~\cite{deSousa2014} to translate the \gls{iec}~61131-3 \gls{pou}s into portable C (one function per \gls{pou}, plus a generated \code{config\_run\_\_} driver that executes one scan); this C is then cross-compiled by the board's standard toolchain (\code{avr-gcc} for AVR boards, \code{arm-none-eabi-gcc} for ARM boards) and linked against a board-specific runtime that maps the located \code{\%I}/\code{\%Q} addresses onto the \code{digitalRead}/\code{digitalWrite}/\code{analogRead}/\code{analogWrite} primitives of the Arduino core. Three observations follow. First, the \emph{semantics that finally execute} are those of C compiled for the target \gls{abi} -- in particular, integer arithmetic is performed at the target's \code{int} width, not at the abstract model's idealized width. Second, the \emph{inputs} are not free: each is produced by a peripheral with a fixed range. Third, none of this is visible at the PLCopen XML level that existing verifiers analyze, which is exactly why the gap is missed today (Fig.~\ref{fig:toolchain}).

\begin{figure*}[htbp]
    \centering
    \resizebox{\textwidth}{!}{%
    \begin{tikzpicture}[font=\small, >={Stealth[length=2.4mm]},
          stg/.style={draw, semithick, rounded corners=2pt, align=center, text width=21mm, minimum height=12mm, inner sep=3pt},
          abs/.style={stg, fill=gray!10},
          dep/.style={stg, fill=orange!12},
          band/.style={rounded corners=3pt, inner sep=3mm},
          bl/.style={font=\footnotesize\bfseries},
          obs/.style={font=\scriptsize, align=left}]
        \node[abs] (ed) {OpenPLC\\Editor};
        \node[abs, right=6mm of ed] (xml) {PLCopen XML\\project};
        \node[dep, right=16mm of xml] (mat) {MATIEC $\to$ C\\one fn / \acrshort{pou}\\+ \code{config\_run\_\_}};
        \node[dep, right=6mm of mat] (tc) {Board toolchain\\\code{avr-gcc} /\\\code{arm-none-\\eabi-gcc}};
        \node[dep, right=6mm of tc] (rt) {Board runtime\\\code{\%I}/\code{\%Q} $\to$\\\code{digital/\\analogRead}};
        \node[dep, right=6mm of rt] (dev) {Device\\(executes one scan)};
        \draw[->] (ed) -- (xml);
        \draw[->] (xml) -- (mat);
        \draw[->] (mat) -- (tc);
        \draw[->] (tc) -- (rt);
        \draw[->] (rt) -- (dev);
        \begin{scope}[on background layer]
          \node[band, fill=blue!7, fit=(ed)(xml)] (bA) {};
          \node[band, fill=orange!7, fit=(mat)(dev)] (bB) {};
        \end{scope}
        \node[bl, blue!55!black, above=1mm of bA.north] {Analyzed by existing verifiers};
        \node[bl, orange!60!black, above=1mm of bB.north] {What actually executes (target \gls{abi} + peripherals)};
        \draw[red!70!black, dashed, very thick]
          ([xshift=8mm, yshift=5mm]xml.north east) -- ([xshift=8mm, yshift=-5mm]xml.south east);
        \node[red!70!black, font=\scriptsize\bfseries, rotate=90, fill=white, inner sep=1pt]
          at ([xshift=8mm]xml.east) {deployment gap};   
    \end{tikzpicture}%
    }
    \caption{Open-source \gls{plc} deployment: OpenPLC Editor → PLCopen~XML → MATIEC C compilation (\cite{deSousa2014}) → cross-compilation + runtime mapping \code{\%I}/\code{\%Q} to Arduino I/O primitives. Verifiers analyze the abstract model (left), but real semantics (integer widths, bounded I/O) appear only post-compilation (right), creating the deployment gap between them.}
    \Description{Two-column diagram comparing the abstract PLC model on the left with the concrete deployed code on the right. The left column is labeled 'Abstract model (verifiers)'. It shows three stacked boxes from top to bottom: 'PLCopen XML', 'MATIEC + POU C code', and 'Model: unbounded ints, no input bounds'. The right column is labeled 'Concrete device (runtime)'. It shows three corresponding stacked boxes from top to bottom: 'OpenPLC Editor', 'GCC toolchain (target \gls{abi})', and 'Device: 16-bit int, ADC 0-\num{1023}, digital 0/1'. An arrow labeled 'compile/link' points from the middle box in the left column to the middle box to the right column. A horizontal bracket labeled 'Deployment gap' spans between the bottom boxes of both columns, highlighting the discrepancy between the abstract model (unbounded ints, no input bounds) and the concrete device (16-bit int, ADC 0-\num{1023}, digital 0/1). The overall flow shows: OpenPLC Editor exports PLCopen XML; MATIEC compiles each POU to portable C, plus a config_run_ driver; and the board toolchain cross-compiles for the target \gls{abi} and links a runtime that maps IEC addresses to Arduino-core primitives (digitalRead, digitalWrite, analogRead, analogWrite).
}
    \label{fig:toolchain}
\end{figure*}
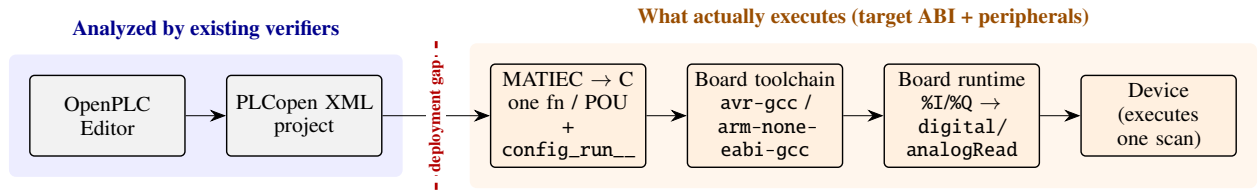

\subsection{Machine Models of Arduino-Class Targets}
\label{sub:machine-models}
The single most consequential difference between the abstract model and the deployed artifact is the \emph{machine word width}, fixed by the target \gls{abi}. On the 8-bit AVR architecture that powers the Arduino Uno, Nano, Mega, and the Arduino-compatible CONTROLLINO~Micro/Maxi and Industrial Shields M-Duino, the C \code{int} is \textbf{16 bits} (range $[-32{,}768,32{,}767]$), following the \code{avr-gcc} \gls{abi}. On the 32-bit ARM Cortex-M cores that power the Arduino Due, Zero, MKR, UNO~R4, and the industrial OPTA, \code{int} is \textbf{32 bits}. Table~\ref{tab:boards} (\S\ref{sec:arduino}) summarizes the relevant parameters. The consequence is stark and routinely surprising to engineers: the multiplication \code{x*100} that is harmless on a Due overflows on a Uno once \code{x} exceeds $327$, because the result leaves the 16-bit range and wraps under two's-complement semantics. Analog peripherals compound this: the AVR \gls{adc} is 10-bit (readings $0\ldots\num{1023}$) and \code{analogWrite} drives an 8-bit \gls{pwm} duty ($0\ldots255$), whereas the UNO~R4 offers a 14-bit \gls{adc}; the \emph{combination} of a narrow word and a wide sensor range is precisely the regime in which scaling arithmetic overflows. We take all of these parameters from the official Arduino cores~\cite{ArduinoCoreAVR} rather than asserting them, so the machine model is auditable and exactly matches what ships on the device.

\subsection{Bounded Model Checking and \texorpdfstring{$k$}{k}-Induction for PLCs}
\label{sub:bmc-primer}
Our backend is \gls{esbmc}~\cite{Cordeiro2012esbmc, gadelha2020}, an \gls{smt}-based \gls{bmc} engine. \gls{bmc} unrolls the scan loop a bounded number of times $k$ and encodes ``does a property violation occur within $k$ scans?'' as an \gls{smt} formula whose satisfying assignment, if any, is a concrete counterexample trace~\cite{gadelha2020}. Because \gls{esbmc} models machine integers as fixed-width \emph{bit-vectors}, arithmetic overflow and two's-complement wrap-around are reproduced \emph{bit-precisely} -- the property of the solver we rely on to detect width-dependent defects, and the reason a faster word-agnostic abstraction would be unsound here. Plain \gls{bmc} is a falsifier: it finds bugs up to depth $k$ but does not prove their absence. 

To obtain unbounded proofs over the non-terminating scan loop, \gls{esbmc} uses \emph{$k$-induction}~\cite{Sheeran2000kinduction, gadelha2020}: a base case checks the first $k$ scans, and an inductive step shows that if the property holds for $k$ consecutive scans it holds for the next, with internal state hacked. For a per-scan, state-independent property such as ``the scaling arithmetic does not overflow,'' the inductive step typically closes at small $k$; for properties whose proof requires reasoning about loop-carried internal state -- as in function blocks with internal counters -- the inductive step may fail to close, yielding an inconclusive \code{unknown} verdict rather than \code{safe}. This distinction explains both why our method proves many programs robust and why a residual set remains \code{unknown} (\S\ref{sec:eval}). Overflow itself is encoded as a verification condition: for a subtraction \code{a-b} of signed type, \gls{esbmc} asserts \code{!overflow("-", a, b)}, whose negation the solver tries to satisfy within the modeled input ranges (Fig.~\ref{fig:kinduction}).

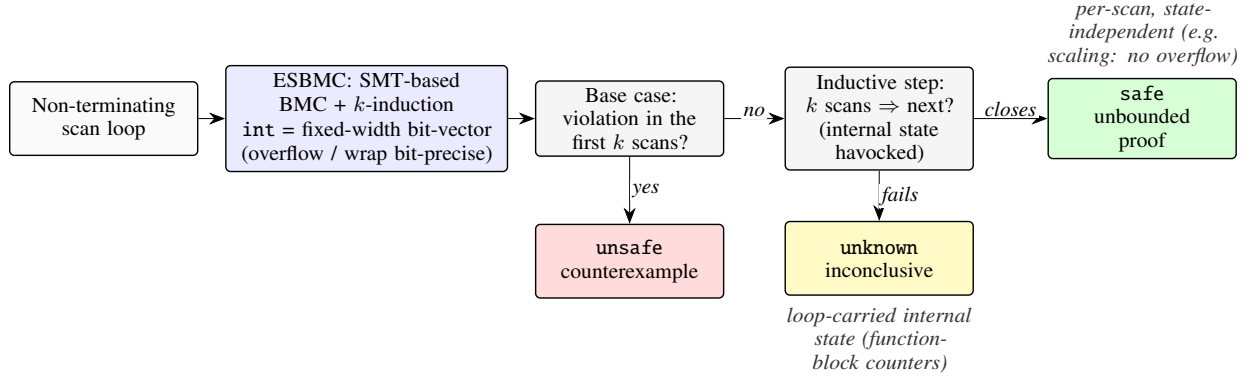
\begin{figure*}[htbp]
    \centering
    \resizebox{\textwidth}{!}{%
    \begin{tikzpicture}[font=\small, >={Stealth[length=2.4mm]},
          box/.style={draw, semithick, rounded corners=2pt, align=center, text width=26mm, inner sep=3pt, minimum height=11mm},
          eng/.style={box, fill=blue!8, text width=40mm},
          dec/.style={box, fill=gray!8},
          vbad/.style={box, fill=red!14, text width=26mm},
          vsafe/.style={box, fill=green!16, text width=26mm},
          vunk/.style={box, fill=yellow!28, text width=26mm},
          el/.style={font=\footnotesize\itshape, fill=white, inner sep=1pt},
          gl/.style={font=\footnotesize\itshape, text=black!80, align=center, text width=32mm}]
        \node[box, fill=gray!4] (loop) {Non-terminating\\scan loop};
        \node[eng, right=4mm of loop] (eng) {ESBMC: \acrshort{smt}-based \acrshort{bmc} + $k$-induction\\[1pt]\footnotesize \code{int} = fixed-width bit-vector (overflow / wrap bit-precise)};
        \node[dec, right=4mm of eng] (base) {Base case:\\violation in the first $k$ scans?};
        \node[dec, right=9mm of base] (ind) {Inductive step:\\$k$ scans $\Rightarrow$ next?\\(internal state havocked)};
        \node[vbad,  below=10mm of base] (bad)  {\code{unsafe}\\counterexample};
        \node[vsafe, right=11mm of ind]  (safe) {\code{safe}\\unbounded\\ proof};
        \node[vunk,  below=7.5mm of ind]  (unk)  {\code{unknown}\\inconclusive};
        \draw[->] (loop) -- (eng);
        \draw[->] (eng) -- (base);
        \draw[->] (base) -- node[el, above]{no} (ind);
        \draw[->] (base) -- node[el, right]{yes} (bad);
        \draw[->] (ind) -- node[el, above]{closes} (safe);
        \draw[->] (ind) -- node[el, right]{fails} (unk);
        \node[gl, above=0.5mm of safe.north] {per-scan, state-independent (e.g.\ scaling: no overflow)};
        \node[gl, below=0.5mm of unk.south]  {loop-carried internal state (function-block counters)};
    \end{tikzpicture}%
    }
    \caption{Verification with \gls{esbmc}: \gls{bmc} unrolls the scan loop \textit{k} times and encodes violations as bit-vector \gls{smt} formulas, capturing overflow and wrap-around exactly. \textit{k}-induction closes state-independent properties at small \textit{k} (\code{safe}) but may leave loop-carried state open (\code{unknown}). Overflow itself is checked as a condition, e.g., \code{!overflow("-", a, b)} for signed subtraction—with the solver seeking inputs within modeled ranges.}    
    \Description{Two-column ESBMC verification diagram. Left: BMC unrolls the scan loop k times, encodes property violations as SMT formulas over fixed-width bit-vectors, producing UNSAFE if a satisfying assignment (counterexample) is found. Right: k-induction with havocked state between windows yields SAFE for state-independent properties or UNKNOWN for loop-carried state. A bottom box shows overflow as a verification condition for signed subtraction.
    }
    \label{fig:kinduction}
\end{figure*}
\section{Related Work}
\label{sec:related}

We situate the work along three axes -- verification of \gls{iec}~61131-3, target-faithful analysis of embedded software, and code-generation correctness -- and summarize the positioning in Table~\ref{tab:related}.

\subsection{Verification of IEC~61131-3}
The most mature open-source effort is \textbf{PLCverif}~\cite{LopezMiguel2022, LopezMiguel2025}, developed at \gls{cern} and in production since~2019; it accepts Siemens \gls{scl} and dispatches to \gls{cbmc}, nuXmv, or Theta, with engineer-facing requirement patterns.

Earlier and complementary academic tools include:
\begin{itemize}
    \item Arcade.PLC~\cite{Weiss2021}, which provides counterexample-guided analysis and value-set abstraction over several \gls{iec}~61131-3 languages;
    
    \item Model-checking frontends that translate \gls{il} or \gls{st} into the input languages of NuSMV/nuXmv~\cite{Cavada2014}.
\end{itemize}

A complementary and fast-growing line establishes precise formal semantics for \gls{iec}~61131-3 and analyzes richer execution settings:
\begin{itemize}
    \item K-ST gives an executable semantics for \gls{st}~\cite{Wang2023};
    
    \item \citet{Lee2024} formalize multitask \gls{st} with preemption;
    
    \item Closest to our cyber-physical concern, \citet{Lee2025} formalize networked \glspl{plc} interacting with physical environments;
    
    \item FRET-based monitor synthesis extends PLCverif with requirement formalization~\cite{Fink2024};
    
    \item Cooperative verification orchestrates multiple backends via CoVeriTeam~\cite{Ukegbu2023a};
    
    \item Ladder logic has been formalized for fault-tolerance analysis~\cite{Ebnenasir2023}.
\end{itemize}

Industrial practice mirrors this -- an \gls{smt}-based ladder-to-model-checker translation is patented for network anomaly detection~\cite{Bruttomesso2024}.

The \textbf{ESBMC-PLC} line targets the \gls{esbmc} backend specifically:
\begin{itemize}
    \item ESBMC-PLC~\cite{DantasCordeiro2026artefact} introduced a native textual PLCopen XML \gls{ld} frontend with \textit{k}-induction;
    
    \item ESBMC-GraphPLC~\cite{DantasCordeiro2026graphical} added the graphical encoding via a \gls{dfs} rung resolver;
    
    \item ESBMC-PLC+~\cite{ESBMCpr5400} added \gls{st}/\gls{scl} ingestion through MATIEC and scan-cycle-accurate semantics for timer, counter, and edge function blocks.
\end{itemize}

What unites \emph{all} of these tools -- commercial and academic alike -- is that they verify an \emph{abstract} model of the control logic: integers are mathematical (or fixed at a host width chosen for convenience, typically 32~bits), and inputs are unconstrained. None models the word width of the deployment target or the bounded domains of the physical inputs. The deployment gap is therefore orthogonal to and composable with each of these efforts; our method could, in principle, be layered onto PLCverif or Arcade.PLC as readily as onto \gls{esbmc}.

\subsection{Target-Faithful Analysis of Embedded Software}
The discipline we transfer to the \gls{plc} domain is well established for handwritten embedded~C.

\begin{itemize}
    \item Frama-C~\cite{Kirchner2015framac} carries an explicit machine model selected with \code{-machdep} (e.g.\ \code{avr} fixes a 16-bit \code{int}), and its value analysis obtains input ranges from ACSL contracts written by the Engineer.
    
    \item Astr\'ee~\cite{Cousot2005astree}, the abstract interpreter used in avionics certification, reliably reports integer and floating-point overflow in the target's arithmetic and derives input bounds from volatile-range and \code{known\_fact} annotations.
    
    \item Industrial static analyzers in the same lineage detect overflow under a configured target model.
\end{itemize}

All of these ``deal with the gap'' in exactly the manner we advocate -- \emph{target machine model plus bounded inputs} -- which is precisely why our approach is defensible rather than speculative.

Two differences distinguish our work.

\begin{enumerate}
    \item The object of analysis is \gls{iec}~61131-3, not C: the Engineer never sees or writes the generated C, so any hardware model must be expressed at, or recovered for, the \gls{plc} level.
    
    \item And centrally, those tools require the input ranges to be \emph{supplied manually} (as contracts or annotations), whereas we \emph{derive} them automatically from the \gls{iec}~61131-3 direct-addressing scheme: a \code{\%IW} location is, by the standard, an analog input word, and the board descriptor fixes its \gls{adc} width. The Engineer writes nothing.
\end{enumerate}
\subsection{Code-Generation Correctness}
A separate body of work establishes that a code generator \emph{preserves semantics} from model to executable:

\begin{itemize}
    \item CompCert~\cite{Leroy2009compcert} for general-purpose C-to-assembly.
    
    \item Qualified generators such as SCADE~KCG in the model-based-design world.
\end{itemize}

These guarantees are real but orthogonal: even a perfectly verified translation does not prevent a program's \emph{own} arithmetic from overflowing on a small target, because the overflow is a property of the source semantics under the target word width, not of the translation. Our method is complementary -- it checks the property that code-generation correctness deliberately does not.
\begin{table}[htbp]
    \centering\small
    \caption{Positioning: who models the deployment target. ``Auto inputs'' = hardware input ranges derived automatically rather than supplied by hand}
    \label{tab:related}
    \begin{tabular}{lcccc}
        \toprule
        \textbf{Approach} & \textbf{IEC 61131-3} & \textbf{Word width} & \textbf{Input ranges} & \textbf{Auto inputs} \\
        \midrule
        PLCverif~\cite{LopezMiguel2022}        & \checkmark & --        & --        & --        \\
        Arcade.PLC~\cite{Weiss2021}            & \checkmark & --        & partial   & --        \\
        ESBMC-PLC+~\cite{ESBMCpr5400}          & \checkmark & --        & --        & --        \\
        CBMC~\cite{Clarke2004cbmc}             & --         & \checkmark & manual    & --        \\
        Frama-C~\cite{Kirchner2015framac}      & --         & \checkmark & contracts & --        \\
        Astr\'ee~\cite{Cousot2005astree}       & --         & \checkmark & annot.\   & --        \\
        \textbf{This work (\arduinotool{})}    & \checkmark & \checkmark & \checkmark & \checkmark \\
        \bottomrule
    \end{tabular}
\end{table}

\section{The Deployment Gap, Formally}
\label{sec:gap}

\subsection{Physical Process and Fault Model}
\label{sub:phys-model}
An open-hardware \gls{plc} is one element of a closed cyber-physical loop (Fig.~\ref{fig:phys}). A physical process -- a tank, a burner, a conveyor -- is observed by \emph{sensors} whose continuous signals are digitized by the \gls{mcu}'s \gls{adc} into integer codes; the scan-cycle logic reads those codes, computes in the \gls{mcu}'s finite-width integer arithmetic, and drives \emph{actuators} through digital pins and \gls{pwm}, which in turn act on the process. A safety property is ultimately a property of this loop -- \emph{``the high-level alarm actuates before the tank overflows''} -- realized in software as an assertion over the program variables that bind to the physical I/O. The cyber and the physical meet at exactly two surfaces: the \gls{adc} that turns a bounded physical quantity into an integer (so each input has a \emph{hardware-realizable range}, not the full machine word), and the integer arithmetic that runs at the target's \emph{word width} (so a computed actuation command can wrap). These two surfaces are the locus of the deployment gap.

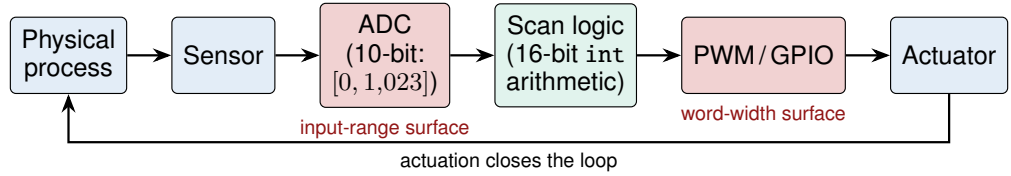
\begin{figure}[htbp]
    \centering
    \resizebox{0.8\linewidth}{!}{%
    \begin{tikzpicture}[
        font=\small\sffamily,
        blk/.style={draw, rounded corners=2pt, align=center, minimum height=0.95cm, inner sep=4pt, fill=codebg},
        phys/.style={blk, fill=codeblue!12},
        cyb/.style={blk, fill=figteal!40},
        surf/.style={blk, fill=unsafrd!20},
        arr/.style={-{Stealth}, thick},
        node distance=0.55cm
        ]
        \node[phys] (plant) {Physical\\process};
        \node[phys, right=of plant] (sensor) {Sensor};
        \node[surf, right=of sensor] (adc) {\gls{adc}\\(10-bit:\\$[0,\num{1023}]$)};
        \node[cyb, right=of adc] (logic) {Scan logic\\(16-bit \code{int}\\arithmetic)};
        \node[surf, right=of logic] (pwm) {\gls{pwm}\,/\,\gls{gpio}};
        \node[phys, right=of pwm] (act) {Actuator};
        \draw[arr] (plant)--(sensor);
        \draw[arr] (sensor)--(adc);
        \draw[arr] (adc)--(logic);
        \draw[arr] (logic)--(pwm);
        \draw[arr] (pwm)--(act);
        \draw[arr] (act.south) -- ++(0,-0.6) -| (plant.south) node[pos=0.25, below, font=\scriptsize\sffamily] {actuation closes the loop};
        \node[font=\scriptsize\sffamily, below=0.02cm of adc, text=unsafrd!80!black] {input-range surface};
        \node[font=\scriptsize\sffamily, below=0.02cm of pwm, text=unsafrd!80!black] {word-width surface};
    \end{tikzpicture}%
    }
    \caption{The cyber-physical loop of an open-hardware \gls{plc}. The deployment gap lies at the two interfaces where computation meets physical I/O: the \gls{adc} bounds each input to a hardware-realizable range (10-bit, $[0,\num{1023}]$), and the scan logic executes 16-bit integer arithmetic at the target word width. Properties verified against the abstract logic alone ignore both, rendering the verification unsound with respect to the actual hardware.}    
    \Description{Block diagram of the PLC cyber-physical feedback loop.  Six blocks are arranged left-to-right: Physical process (light blue), Sensor (light blue), ADC (green, labeled '10-bit: [0,\num{1023}]'), Scan logic (teal, labeled '16-bit int arithmetic'), PWM/GPIO (green), and Actuator (light blue).  Solid arrows connect them sequentially.  A feedback arrow curves from the bottom of the Actuator down and left to the bottom of the Physical process, labeled 'actuation closes the loop'.  Two shaded surfaces are highlighted: one below the ADC block (input-range surface) and one below the PWM block (word-width surface) – these represent the deployment gap that mathematical verification must account for, including real hardware constraints such as ADC ranges and integer word widths.}
    \label{fig:phys}
\end{figure}

\subsubsection{Fault model} We do \emph{not} assume a malicious adversary; the fault we target is a \emph{latent computational fault induced by legitimate physical inputs}. A program that is correct in the abstract model can, on a narrow-word target, compute the wrong actuation for a perfectly ordinary sensor reading -- an integer overflow in an \gls{adc}-to-engineering-unit conversion silently inverts a comparison and suppresses a safety action. The fault is triggered by inputs the environment routinely produces (within $D(a)$, the bounded domain $D$ of physically realizable values for analog input $a$), is invisible at the \gls{plc}-logic level existing verifiers analyze, and manifests only once the logic is composed with a specific board's \gls{hal}. Dually, an \emph{unconstrained} input model fabricates faults the environment can never trigger (inputs outside $D(a)$): these are the phantom alarms of \S\ref{sec:eval}. A faithful model must therefore admit exactly the physically realizable inputs -- no more, no less.

\subsubsection{Assumptions} We assume the standard \gls{plc} scan-cycle timing model (\S\ref{sub:exec-model}): inputs are sampled at the top of each scan and held constant during rung evaluation, so we reason about values, not intra-scan transients, and we do not model continuous plant dynamics, sensor noise, or actuator latency -- the property is the per-scan logical assertion, and the environment is the open-world set of realizable input vectors. Within this model, the gap is purely a function of the program, the property, and the board \gls{hal} $H$, which is what makes it amenable to the bit-precise analysis that follows.

Let $P$ be a program and $\varphi$ a safety property. Existing tools check $M_{\text{abs}}(P)\models\varphi$, where $M_{\text{abs}}$ is the abstract scan-cycle model (\S\ref{sub:exec-model}) with idealized arithmetic (word width $w\to\infty$) and free inputs ($D(i)=\mathbb{Z}$). What runs on the board is $M_{\text{dev}}(P,H)$: the same control logic composed with a board \gls{hal} $H$ that (i)~evaluates arithmetic in the \gls{mcu}'s finite word width $w$ and (ii)~constrains each input to its hardware-realizable domain $D(i)$.

\begin{definition}[Deployment gap]
    The deployment gap of $(P,\varphi)$ on target $H$ is the situation in which the two models disagree on $\varphi$: either $M_{\text{abs}}(P)\models\varphi$ while $M_{\text{dev}}(P,H)\not\models\varphi$ (a \emph{missed defect}: safe in the abstract, unsafe on the device), or $M_{\text{abs}}(P)\not\models\varphi$ under free inputs while $M_{\text{dev}}(P,H)\models\varphi$ (a \emph{phantom}: an abstract counterexample that no hardware input can realize). The first direction motivates detection; the second, as RQ2 shows, dominates in practice and renders naive width-checking unsound (Fig.~\ref{fig:gapcases}).
\end{definition}

\begin{figure}[htbp]
    \centering
    \resizebox{0.85\columnwidth}{!}{%
    \begin{tikzpicture}[font=\small,
          cell/.style={draw, semithick, align=center, text width=36mm, minimum height=16mm, inner sep=3pt},
          agree/.style={cell, fill=gray!10},
          miss/.style={cell, fill=red!14},
          phan/.style={cell, fill=yellow!30},
          colh/.style={font=\footnotesize\bfseries, align=center, text width=36mm},
          rowh/.style={font=\footnotesize\bfseries, align=center, text width=17mm}]
        \node[agree] (tl) {agree:\\safe};
        \node[miss, right=1.5mm of tl] (tr) {\textbf{missed defect}\\[1pt]\footnotesize safe in the abstract,\\unsafe on the device};
        \node[phan, below=1.5mm of tl] (bl) {\textbf{phantom}\\[1pt]\footnotesize abstract counterexample\\no hardware input realizes};
        \node[agree, below=1.5mm of tr] (br) {agree:\\genuine defect};
        \node[colh, above=2mm of tl] {$M_{\mathrm{dev}}(P,H)\models\varphi$\\[1pt]\footnotesize safe on device};
        \node[colh, above=2mm of tr] {$M_{\mathrm{dev}}(P,H)\not\models\varphi$\\[1pt]\footnotesize unsafe on device};
        \node[rowh, left=3mm of tl] {$M_{\mathrm{abs}}(P)\models\varphi$\\[1pt]\footnotesize safe in abstract};
        \node[rowh, left=3mm of bl] {$M_{\mathrm{abs}}(P)\not\models\varphi$\\[1pt]\footnotesize free inputs};
        \node[font=\scriptsize\itshape, text=red!60!black, right=1.5mm of tr, anchor=west, text width=22mm]
          {motivates\\ detection};
        \node[font=\scriptsize\itshape, text=black!65, below=1mm of bl, text width=36mm, align=center]
          {dominates in practice (RQ2~\S\ref{sec:eval}) -- makes naive width-checking unsound};
    \end{tikzpicture}%
    }
    \caption{Deployment gap as model agreement (diagonal) vs.\ disagreement (anti-diagonal) on $\varphi$. \textit{Missed defects} (top-right) are device-unsafe but abstract-safe -- requiring detection. \textit{Phantoms} (bottom-left) are abstract counterexamples unreachable in hardware, hence device-safe. RQ2 (\S\ref{sec:eval}) shows phantoms dominate, making naive width-checking unsound. Diagonal cells indicate genuine agreement.}
    \Description{2x2 matrix showing agreement (diagonal) and disagreement (anti-diagonal) between abstract and device model verdicts. Top-left: SAFE/SAFE -- genuinely safe. Bottom-right: UNSAFE/UNSAFE -- genuine defect. Top-right: SAFE/UNSAFE -- missed defect (abstract misses a real hardware fault). Bottom-left: UNSAFE/SAFE -- phantom (abstract counterexample unreachable on real hardware). Phantoms dominate in practice, making naive width checking unsound.
}
    \label{fig:gapcases}
\end{figure}

Three mechanisms populate the gap.

\begin{enumerate}
  \item \textbf{(M1)~Binding errors}: an \code{\%IX}/\code{\%QX} address mapped to the wrong physical pin, or two addresses aliased to one pin, so that an output write clobbers an unrelated signal -- a fault in $\mu$ rather than in the logic, modeled by the \code{io\_map} and detectable as a property over the bound variables.
  
  \item \textbf{(M2)~\gls{adc} quantization}: a physical level maps to one of only $2^{b_{\mathrm{adc}}}$ discrete codes, so a threshold test such as ``alarm exactly at \SI{50.0}{\percent}'' may straddle two representable codes and never hold with equality; we model a read as a nondeterministic representable value, a sound over-approximation of any rounding policy.
  
  \item \textbf{(M3)~Bounded-width arithmetic}: an operation that is exact in $M_{\text{abs}}$ exceeds $[-2^{w-1},2^{w-1})$ and wraps under two's-complement semantics in $M_{\text{dev}}$. Formally, an operation $\odot$ on $w$-bit operands computes $a\odot b \bmod 2^{w}$ reinterpreted as signed, so $a\odot b$ is faithful iff $a\odot b\in[-2^{w-1},2^{w-1})$; outside that interval the device silently produces a different value. (M3) is the dominant cause of overflow defects and the focus of this paper; (M1)~and~(M2) are expressible in the same framework and are future work (\S\ref{sec:future}).
\end{enumerate}

\subsubsection{Motivating example} Listing~\ref{lst:gap} shows an analog high-level alarm from the OpenPLC water/dimmer pattern: a 10-bit \gls{adc} reading is scaled to a percentage, and an alarm fires at $\geq 80\%$. The property is \emph{``if the sensor is at/above the high mark ($adc\geq819$) the alarm must be ON.''} Under $M_{\text{abs}}$ (32-bit) \gls{esbmc} proves it ($k{=}2$). Under $M_{\text{dev}}$ for an AVR board, where \code{int} is 16-bit, the product $adc\times100$ wraps for high readings; \gls{esbmc} reports a violation with counterexample $adc=898$ (tank ${\approx}88\%$), where the alarm stays silent -- a missed high-level alarm caused purely by the word width. The two models differ only in the \gls{hal}-injected arithmetic.
\begin{minipage}{\textwidth}
\begin{lstlisting}[language=C, caption={Deployment-gap motivating example. Abstract \code{int}~(32-bit)$\Rightarrow$SAFE; deployed \code{int16\_t}~(AVR)$\Rightarrow$missed alarm at $adc=898$.}, label={lst:gap}]
    int adc = nondet_int();          // HAL: 10-bit ADC
    __ESBMC_assume(adc >= 0 && adc <= 1023);
    int16_t prod    = adc * 100;     // HAL-faithful width on AVR
    int16_t percent = prod / 1023;   // wraps for high adc
    bool    Alarm   = (percent >= 80);
    if (adc >= 819) assert(Alarm);   // engineer's safety property
\end{lstlisting}
\end{minipage}
\section{Hardware-Faithful Verification}
\label{sec:method}

Given a program $P$ and a \gls{hal} descriptor $H$, the method constructs a \emph{deployment model} $M_{dev}(P,H)$. It verifies the safety property $\varphi$ against it, reusing the \gls{esbmc} backend and \textit{k}-induction unchanged (Fig.~\ref{fig:arch}). We first fix the underlying execution model, then define $H$, the lowering, and its soundness.

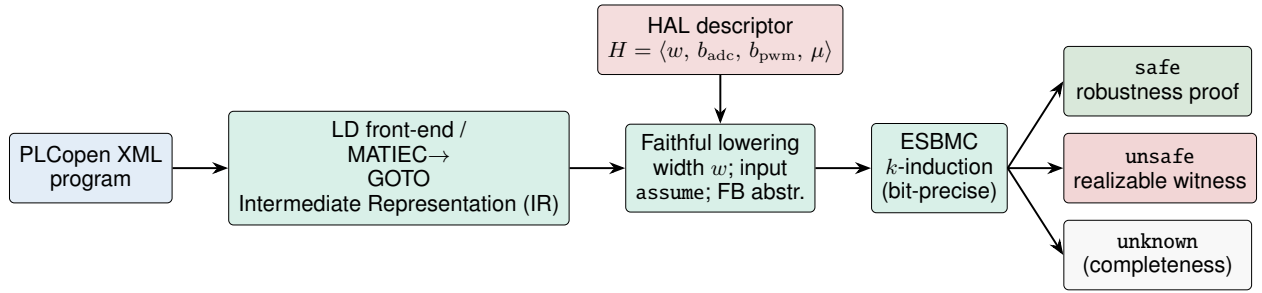
\begin{figure}[htbp]
    \centering
    \resizebox{\linewidth}{!}{%
    \begin{tikzpicture}[
        font=\small\sffamily,
        box/.style={draw, rounded corners=2pt, align=center, minimum height=1.0cm, fill=codebg, inner sep=4pt},
        io/.style={box, fill=codeblue!12},
        proc/.style={box, fill=figteal!50},
        hal/.style={box, fill=unsafrd!15},
        verdict/.style={box, minimum width=2.7cm},
        arr/.style={-{Stealth}, thick},
        node distance=0.65cm and 0.8cm
        ]
        \node[io] (xml) {PLCopen XML\\program};
        \node[proc, right=of xml] (fe) {LD front-end /\\MATIEC$\to$\\ GOTO\\ \gls{ir}};
        \node[proc, right=of fe] (low) {Faithful lowering\\width $w$;\ input\\ \code{assume};\ FB abstr.};
        \node[proc, right=of low] (esbmc) {\gls{esbmc}\\$k$-induction\\(bit-precise)};
        \node[hal, above=0.7cm of low] (hal) {\gls{hal} descriptor\\ $H=\langle w,\,b_{\mathrm{adc}},\,b_{\mathrm{pwm}},\,\mu\rangle$};
        \node[verdict, right=of esbmc, fill=codegreen!18, yshift=1.25cm] (safe) {\code{safe}\\robustness proof};
        \node[verdict, right=of esbmc, fill=unsafrd!18] (unsafe) {\code{unsafe}\\realizable witness};
        \node[verdict, right=of esbmc, fill=codebg, yshift=-1.25cm] (unk) {\code{unknown}\\(completeness)};
        \draw[arr] (xml)--(fe);
        \draw[arr] (fe)--(low);
        \draw[arr] (hal)--(low);
        \draw[arr] (low)--(esbmc);
        \draw[arr] (esbmc.east) -- (safe.west);
        \draw[arr] (esbmc.east) -- (unsafe.west);
        \draw[arr] (esbmc.east) -- (unk.west);
    \end{tikzpicture}%
    }
    \caption{\arduinotool{} pipeline. A PLCopen program is lowered to a GOTO \gls{ir}; the faithful lowering (\S\ref{sub:lowering}), parameterized by the board's \gls{hal} descriptor, fixes the word width and injects hardware input bounds before bit-precise \textit{k}-induction. Setting $w$ wide yields the abstract reference $M_{\text{abs}}$; the per-board $w$ yields $M_{\text{dev}}$. The backend is reused unchanged.}
    \Description{Block diagram of the PLC verification pipeline. The flow proceeds left to right. On the left, a light blue block labeled 'PLCopen XML program' points to a teal block labeled 'LD frontend / MATIEC → GOTO IR'. This points to a second teal block labeled 'Faithful lowering width w; input assume; FB abstr.'. A light orange block labeled '\gls{hal} descriptor H = (w, b_adc, b_pwm, mu)' is positioned above the lowering block and points down into it. The lowering block points to a fourth teal block labeled 'ESBMC k-induction (bit-precise)'. On the far right, three verdict blocks are stacked vertically: the top one is light green and labeled 'SAFE robustness proof', the middle one is light red and labeled 'UNSAFE realizable witness', and the bottom one is gray and labeled 'UNKNOWN (completeness)'. Three separate arrows point from the ESBMC block to each of these three verdict blocks.
    }
    \label{fig:arch}
\end{figure}

\subsection{Scan-Cycle Execution Model}
\label{sub:exec-model}
We model a \gls{plc} program as a transition system $\mathcal{T}=(S,s_0,R)$ over a state $S$ that partitions variables into inputs $\mathcal{I}$, outputs $\mathcal{O}$, and internal state $\mathcal{V}$. One scan is the relation $\mathrm{scan}\subseteq S\times S$ that (i)~re-samples every input $i\in\mathcal{I}$, (ii)~executes the program body once, updating $\mathcal{O}$ and $\mathcal{V}$, and the system iterates $\mathrm{scan}$ indefinitely. A safety property $\varphi$ is an assertion required to hold after every scan. Two parameters of this model are left implicit by existing verifiers but fixed on a real device: the \emph{interpretation of arithmetic operators} (the word width $w$ at which $+,-,\times$ are evaluated, with two's-complement wrap outside $[-2^{w-1},2^{w-1})$) and the \emph{input domains} (the set $D(i)\subseteq\mathbb{Z}$ from which each $i\in\mathcal{I}$ may be sampled). The abstract model $M_{\text{abs}}$ takes $w$ large and $D(i)=\mathbb{Z}$; the deployment model $M_{dev}(P,H)$ takes both from $H$ (Fig.~\ref{fig:scan-cycle}.

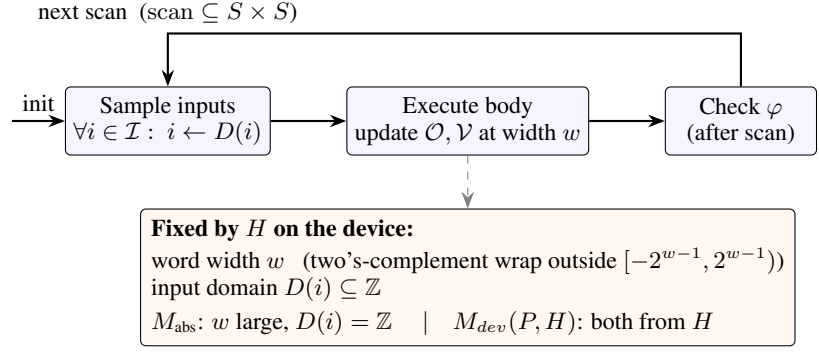
\begin{figure}[htbp]
    \centering
    \begin{tikzpicture}[
        font=\small,
        >={Stealth[length=2.2mm]},
        node distance=6mm,
        stage/.style={
            draw, rounded corners=2pt, align=center,
            minimum height=9mm, minimum width=20mm,
            inner sep=3pt, fill=blue!4
        },
        param/.style={
            draw, rounded corners=2pt, align=left,
            inner sep=4pt, fill=orange!6, font=\footnotesize
        },
        flow/.style={->, thick}
    ]
    
    \node[stage] (sample) {Sample inputs\\$\forall i\in\mathcal{I}:\ i \leftarrow D(i)$};
    \node[stage, right=10mm of sample] (body) {Execute body\\update $\mathcal{O},\mathcal{V}$ at width $w$};
    \node[stage, right=10mm of body] (check) {Check $\varphi$\\(after scan)};
    
    \draw[flow] ($(sample.west)+(-7mm,0)$) -- node[above, font=\footnotesize] {init} (sample.west);
    \draw[flow] (sample) -- (body);
    \draw[flow] (body) -- (check);
    
    \draw[flow]
        (check.north) -- ++(0,7mm)
        -| node[pos=0.5, above, font=\footnotesize] {next scan \ ($\mathrm{scan}\subseteq S\times S$)}
        (sample.north);
    
    \node[param, below=7mm of body, xshift=0mm] (params) {%
        \textbf{Fixed by $H$ on the device:}\\[2pt]
        word width $w$ \ \ (two's-complement wrap outside $[-2^{w-1},2^{w-1})$)\\
        input domain $D(i)\subseteq\mathbb{Z}$\\[3pt]
        $M_{\text{abs}}$: $w$ large, $D(i)=\mathbb{Z}$ \quad$\mid$\quad
        $M_{dev}(P,H)$: both from $H$%
    };
    \draw[->, dashed, gray] (body.south) -- (params.north -| body.south);
    
    \end{tikzpicture}
    \caption{The scan-cycle execution model. Each scan resamples the inputs, executes the program body once, and ensures that the safety property $\varphi$ holds afterward; the system iterates $\mathrm{scan}$ indefinitely. Two parameters left implicit by existing verifiers are fixed on a real target: the arithmetic word width $w$ and the per-input domains $D(i)$. The abstract model $M_{\text{abs}}$ leaves both unconstrained; the deployment model $M_{dev}(P,H)$ takes both from the \gls{hal} descriptor~$H$.}
    \Description{Block diagram of the \gls{plc} scan-cycle execution model. A loop cycles through three phases: ``Sample inputs'' (reads sensors), ``Execute program body once" (runs scan logic), and Property check: $\varphi$ must hold'' (verifies safety). An arrow loops back from the property check to sampling, labeled ``iterates indefinitely''. Two parameter boxes are shown: ``Word width w'' and ``Input domains D(i)'' — these are fixed by hardware but left unconstrained in the abstract model.}
    \label{fig:scan-cycle}
\end{figure}

\subsection{The \acrfull{hal} Descriptor}
A \gls{hal} descriptor is a tuple
$H=\langle w,\,b_{\mathrm{adc}},\,b_{\mathrm{pwm}},\,\mu\rangle$ where $w$ is the machine word width (\code{int\_bits}), $b_{\mathrm{adc}}$ and $b_{\mathrm{pwm}}$ are the \gls{adc} and \gls{pwm} resolutions in bits, and $\mu$ (the \code{io\_map}) assigns each \gls{iec} address class a role and domain (Eq.~\ref{eq:hal}).
\begin{equation}\label{eq:hal}
    \mu(\code{\%IX})=\langle\text{in},\{0,1\}\rangle,\quad
    \mu(\code{\%IW})=\langle\text{in},[0,2^{b_{\mathrm{adc}}}\!-\!1]\rangle,\quad
    \mu(\code{\%QX}),\mu(\code{\%QW})=\text{out}.
\end{equation}

$H$ is small, declarative, and reusable across all programs for a given board; Table~\ref{tab:boards} lists the instances we use. Crucially, $H$ is \emph{independent of $P$}: the per-program input bounds are obtained by composing $H$ with the addresses already present in $P$.

\subsection{Faithful Lowering}
\label{sub:lowering}
The lowering $P\mapsto M_{dev}(P,H)$ is the identity on the control logic and modifies only the hardware boundary, via four transformations (Algorithm~\ref{alg:lower}):

\begin{enumerate}[leftmargin=2em]
    \item \textbf{Word width.} All machine integers are interpreted as $w$-bit bit-vectors, so every arithmetic operator wraps exactly as on the target. In \gls{esbmc}, this is the word-width setting for the architecture; no program rewriting is required.
    
    \item \textbf{Input-range annotation.} For each input $i\in\mathcal{I}$ with address $a$, immediately after $i$ is re-sampled in $\mathrm{scan}$ we inject $\mathtt{assume}\big(i\in D(a)\big)$ with $D(a)$ the domain from $\mu$. Booleans receive $\{0,1\}$ (already implied by their type and thus a no-op); analog/integer inputs receive $[0,2^{b_{\mathrm{adc}}}-1]$. This single step is what restores soundness in practice.
    
    \item \textbf{Quantization and saturation.} An \gls{adc} read is modeled as a nondeterministic value within the representable set -- a sound over-approximation that subsumes any rounding policy -- and a \gls{pwm} write is clamped to $[0,2^{b_{\mathrm{pwm}}}-1]$.
    
    \item \textbf{Non-linear abstraction.} A function block whose body uses floating-point or transcendental operations (e.g.\ \gls{pid}, thermistor linearization) is replaced by a fresh nondeterministic output constrained to its declared range. This is a sound over-approximation that removes \gls{smt}-intractable arithmetic outside the integer deployment-gap claim.
\end{enumerate}

\begin{algorithm}[htbp]
    \caption{Faithful lowering $P,H \mapsto M_{dev}$}
    \label{alg:lower}
    \begin{algorithmic}[1]
        \State set machine word width to $H.w$ \Comment{bit-precise arithmetic}
        \For{each input $i$ with address $a$ sampled in $\mathrm{scan}$}
        \State emit $i \gets \mathtt{nondet}()$ \Comment{open-world re-sample (existing)}
        \If{$\mu(a)$ is analog/integer} emit $\mathtt{assume}(0 \le i \le 2^{H.b_{\mathrm{adc}}}-1)$ \EndIf
        \EndFor
        \For{each analog output write $o \gets e$}
        \State emit $o \gets \mathtt{clamp}(e, 0, 2^{H.b_{\mathrm{pwm}}}-1)$
        \EndFor
        \For{each non-linear function-block instance $f$}
        \State replace body of $f$ by $\mathtt{out}_f \gets \mathtt{nondet\_in\_range}()$
        \EndFor
        \State \Return resulting model, to be checked with overflow VCs + $k$-induction
    \end{algorithmic}
\end{algorithm}

\subsection{Soundness and Gap Diagnosis}
\label{sub:soundness}
Let $\mathrm{Reach}_{\mathrm{dev}}$ be the states reachable on the physical device (arithmetic at width $w$, inputs from their true hardware ranges $D(a)$), and $\mathrm{Reach}_{M_{dev}}$ those of the constructed model. The soundness argument assumes the descriptor is \emph{faithful}: $H.w$ equals the target's \code{int} width and, for every input address $a$, $D(a)$ contains the true hardware range of that input (Fig.~\ref{fig:soundness}). Under this premise the lowering preserves all device behavior:

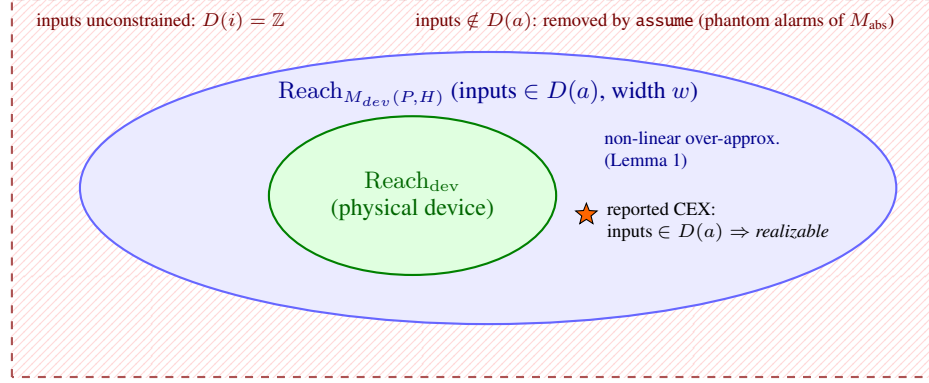
\begin{figure}[htbp]
    \centering
    \begin{tikzpicture}[font=\small, >={Stealth[length=2.2mm]}]
        
        \fill[pattern=north east lines, pattern color=red!18, opacity=0.9]
            (0,0) rectangle (12.2,5.0);
        \draw[draw=red!35!gray, dashed, thick] (0,0) rectangle (12.2,5.0);
        \node[red!45!black, anchor=north west, align=left, font=\scriptsize] at (0.2,4.95)
            {inputs unconstrained: $D(i)=\mathbb{Z}$};
        \node[red!45!black, anchor=north east, align=right, font=\scriptsize] at (11.8,4.95)
            {inputs $\notin D(a)$: removed by \texttt{assume} (phantom alarms of $M_{\text{abs}}$)};
        
        \fill[blue!8] (6.3,2.5) ellipse (5.4cm and 1.8cm);
        \draw[blue!60, thick] (6.3,2.5) ellipse (5.4cm and 1.8cm);
        \node[blue!55!black, anchor=north, align=center, font=\footnotesize] at (6.3,4.05)
            {$\mathrm{Reach}_{M_{dev}(P,H)}$ (inputs $\in D(a)$, width $w$)};
        
        \fill[green!12] (5.3,2.4) ellipse (1.9cm and 1.05cm);
        \draw[green!50!black, thick] (5.3,2.4) ellipse (1.9cm and 1.05cm);
        \node[green!40!black, align=center, font=\footnotesize] at (5.3,2.4)
            {$\mathrm{Reach}_{\mathrm{dev}}$\\(physical device)};
        
        \node[align=left, font=\scriptsize, text=blue!55!black] at (9.0,3.0)
            {non-linear over-approx.\\(Lemma~\ref{lem:nonlin})};
        
        \node[star, star points=5, star point ratio=2.3, fill=orange!80!red,
              draw=black, inner sep=1.3pt] (cex) at (7.6,2.15) {};
        \node[anchor=west, align=left, font=\scriptsize] at (7.75,2.05)
            {reported CEX:\\inputs $\in D(a) \Rightarrow$ \emph{realizable}};
        
    \end{tikzpicture}
    \caption{Soundness as state-set containment (Theorem~\ref{thm:sound}). Under a faithful descriptor~$H$, the states reachable on the physical device are contained in those of the model, $\mathrm{Reach}_{\mathrm{dev}}\subseteq\mathrm{Reach}_{M_{dev}(P,H)}$; the only slack is the non-linear over-approximation (Lemma~\ref{lem:nonlin}). Hence $M_{dev}(P,H)\models\varphi$ implies the device satisfies $\varphi$. Because the injected \texttt{assume}s admit only inputs in $D(a)$, every reported counterexample (orange star) is physically realizable; the hatched band of states reachable only with inputs $\notin D(a)$ -- the phantom alarms of the unconstrained $M_{\text{abs}}$ -- is excluded.}
    \Description{Venn diagram of state-set containment. Reach_dev (physical device states) is a subset of Reach(M_dev(P,H)) (model states) — proving soundness. A hatched band shows states reachable only with inputs outside D(a), excluded by the assume constraints. An orange star marks a physically realizable counterexample.}
    \label{fig:soundness}
\end{figure}

\begin{theorem}[Soundness]
\label{thm:sound}
    For a \gls{hal} descriptor $H$ that correctly captures the target's word width and input domains, $\mathrm{Reach}_{\mathrm{dev}}\subseteq\mathrm{Reach}_{M_{dev}(P,H)}$. Consequently, if $M_{dev}(P,H)\models\varphi$ then the deployed program satisfies $\varphi$ on the device (modulo the non-linear abstraction). Every counterexample reported on $M_{dev}(P,H)$ uses inputs in $D(a)$ and is therefore \emph{physically realizable}.
\end{theorem}

The argument is a containment: width-$w$ arithmetic is reproduced exactly (not abstracted), each non-linear function block is replaced by a nondeterministic over-approximation of its declared outputs, and the injected \code{assume}s remove only inputs the device never produces. So no device behavior is excluded, and every counterexample is realizable. The supporting lemma and full proof are deferred to Appendix~\ref{app:soundness}.

The dual -- that the abstract model is \emph{unsound for the device} -- also follows: without the \code{assume}s, $M_{\text{abs}}$ admits inputs outside $D(a)$, so an $M_{\text{abs}}$ counterexample need not be realizable (Scenario~2, \S\ref{sub:wt-fp}). \textbf{Gap diagnosis} compares the two models on the same program and property with the \emph{same} input bounds, so they differ only in $w$: a verdict that is \code{safe} under $M_{\text{abs}}$ but \code{unsafe} under $M_{dev}$ isolates a defect attributable solely to the deployment word width, reported with a realizable witness. \textbf{Completeness} is inherited from the backend and is partial: \gls{bmc} falsifies up to depth $k$ and $k$-induction proves when the inductive step closes; for programs whose proof depends on loop-carried internal state the step may not close, yielding \code{unknown} (\S\ref{sec:eval}) -- an incompleteness of the prover, independent of the lowering.

\section{\arduinotool{}: Arduino Instantiation and Implementation}
\label{sec:arduino}

\subsection{Arduino Instantiation}
\label{sub:arduino-inst}
We instantiate the method for Arduino. The \gls{hal} library fixes $H$ per board with parameters taken from the official Arduino cores~\cite{ArduinoCoreAVR}: AVR (\code{ArduinoCore-avr}) gives \code{int\_bits}=16, \code{adc\_bits}=10, \code{pwm\_bits}=8; 32-bit references (Due, UNO~R4, OPTA) give \code{int\_bits}=32 with 12--14-bit \glspl{adc} (Table~\ref{tab:boards}). Our primary focus is the 16-bit AVR boards (Uno/Nano/Mega), the dominant low-cost open-hardware \glspl{plc}. The pipeline ingests a PLCopen program through the native \gls{ld} frontend, applies the faithful lowering for the chosen board, and verifies with \gls{esbmc} at the target word width; the reference run uses the abstract model (wide integers, same input bounds).

\begin{table}[htbp]
    \centering\small
    \caption{Board \gls{hal} descriptors (parameters from the official Arduino cores)}
    \label{tab:boards}
        \begin{tabular}{llccc}
        \toprule
        \textbf{Board} & \textbf{Core / MCU} & \textbf{\code{int}} & \textbf{ADC} & \textbf{PWM} \\
        \midrule
        Uno/Nano/Mega, CONTROLLINO, M-Duino & -avr (AVR) & 16-bit & 10 & 8 \\
        Due & -sam (ARM) & 32-bit & 12 & 8 \\
        UNO R4 & -renesas (ARM) & 32-bit & 14 & 8 \\
        OPTA & -mbed (ARM) & 32-bit & 12 &  --  \\
        \bottomrule
    \end{tabular}
\end{table}

\subsection{Implementation}
\label{sec:impl}
A design goal of the method is to be \emph{minimally invasive}: the deployment model should reuse a mature verifier rather than reimplement one. We therefore realize the lowering inside the existing \gls{esbmc} \gls{ld} frontend, changing only where the model meets the hardware boundary.

\subsubsection{Where the Hook Sits}
The \gls{ld} frontend builds a GOTO \gls{ir} in which one scan is a loop body that first \emph{re-samples} the inputs (the open-world sensor model of \S\ref{sub:bmc-primer}) and then evaluates the rungs. Inputs are re-sampled at two sites: the program-level \code{READ\_INPUTS} step, which assigns each physical input a fresh nondeterministic value, and the function-block input-sampling step, which does the same for an instantiated block's formal inputs. The annotator attaches at exactly these two sites, immediately \emph{after} each nondeterministic assignment, so that the bound applies to the freshly sampled value on every scan.

\subsubsection{What the Hook Emits}
For an input variable of integer type, the hook synthesizes a guard $0\le v\le \mathtt{ADCMAX}$ as a relational \gls{ir} expression. It wraps it in a \code{assume} node inserted into the scan body (Listing~\ref{lst:hook}). The bound is taken from the board descriptor (\code{HAL\_ADC\_MAX}, default \num{1023} for the AVR 10-bit \gls{adc}); Boolean inputs already range over $\{0,1\}$ by type and need no guard, and floating-point inputs are left unconstrained (they fall under the non-linear abstraction). Gating on the static type keeps the transformation sound for mixed-type programs.

\begin{lstlisting}[language=C, basicstyle=\ttfamily\scriptsize, caption={The annotator hook (abridged): after an input is re-sampled, bind it to the ADC domain.}, label={lst:hook}]
codet ld_converter::hal_input_bound(const symbol_exprt &in) const {
    if (getenv("HAL_ANNOTATE") == nullptr) return code_skipt(); // baseline
    if (!is_integer(in.type()))            return code_skipt(); // BOOL/REAL
    long adcmax = env_long("HAL_ADC_MAX", 1023);                // board param
    exprt lo = binary_relation_exprt(in, ">=", from_integer(0, in.type()));
    exprt hi = binary_relation_exprt(in, "<=", from_integer(adcmax, in.type()));
    return code_assumet(and_exprt(lo, hi));                     // assume bound
    }
\end{lstlisting}

\subsubsection{Configurability and Engineering}
The model is toggled by \code{HAL\_ANNOTATE} and parameterized by \code{HAL\_ADC\_MAX}, so a \emph{single binary} reproduces both the naive baseline and hardware-faithful configuration -- the basis for the controlled before/after in \S\ref{sec:eval}. Word width is set by \gls{esbmc}'s \code{--16}/\code{--32} flag, and overflow is governed by the standard verification condition; neither requires modification. The hook is $\approx$40~lines, confined to the frontend, and adds no dependencies: the GOTO encoding, \gls{smt} backend, \textit{k}-induction engine, and solver are unchanged. This confinement lets Theorem~\ref{thm:sound} rest on the backend's existing soundness -- we add assumptions and set a width, but do not alter the proof engine. The same hook bounds function-block inputs, including \code{IN1} in the worked example (\S\ref{sub:wt-fp}).
\section{The Approach in Practice}
\label{sec:walkthrough}

We walk through two end-to-end scenarios on real artifacts. Each is presented in two passes: an \emph{Engineer's view} (what happens and what to do about it), followed by an \emph{under the hood} paragraph with the formal detail. The two scenarios cover the two faces of the deployment gap: detecting a genuine hardware-induced defect (\S\ref{sub:wt-detect}), and \emph{not} raising a false alarm on correct code (\S\ref{sub:wt-fp}) -- the latter being, in practice, the more common and the reason the input model is essential.

\subsection{Scenario 1: A Program that is Safe in the Editor but Unsafe on a Uno}
\label{sub:wt-detect}

\subsubsection{Engineer's View} Consider the analog high-level alarm of Listing~\ref{lst:gap}: a tank level is read from a 10-bit \gls{adc}, scaled to a percentage, and an alarm is asserted to fire at $\geq 80\%$. In the OpenPLC editor, and under any abstract verifier, the logic is correct -- the alarm always fires when the tank is high. Deployed to an \textbf{Arduino Uno} (an AVR board), however, the alarm \emph{silently fails} at high levels. \arduinotool{} reports the violation and, crucially, hands back the exact sensor reading that triggers it:
\begin{minipage}{\textwidth}
\begin{lstlisting}[language={}, basicstyle=\ttfamily\scriptsize, caption={Running the deployment check; the counterexample is a real ADC reading.}, label={lst:cmd-detect}]
    # AVR / Uno target (16-bit), HAL input model enabled
    $ HAL_ANNOTATE=1 esbmc alarm.ld --16 --overflow-check --k-induction
      Counterexample: adc = 898
      Violated property: arithmetic overflow on mul (adc * 100)
      VERIFICATION FAILED
    # 32-bit reference (e.g., OPTA), same input bounds
    $ HAL_ANNOTATE=1 esbmc alarm.ld --32 --overflow-check --k-induction
      VERIFICATION SUCCESSFUL
\end{lstlisting}
\end{minipage}
\noindent The engineer reads this directly: \emph{at a tank reading of 898 ($\approx 88\%$) the alarm does not fire on a 16-bit board}. The physical consequence closes the loop of Fig.~\ref{fig:phys}: the high-level alarm that should trip the inflow never actuates, so the tank keeps filling past the high mark -- a safety action silently lost to a word-width overflow, not to any logic error the Engineer could see in the editor. The fix is equally concrete -- use a wider intermediate (\code{DINT}) for the scaling, or deploy on a 32-bit board -- and re-running confirms the repair.

\subsubsection{Under the Hood} The two runs are the deployment model $M_{dev}$ and the abstract reference $M_{\text{abs}}$ of \S\ref{sec:method}, differing only in the machine word width. The \gls{hal} annotator has constrained the input with \code{assume(adc >= 0 \&\& adc <= \num{1023})}, so the counterexample is \emph{realizable}: at the 16-bit width $adc\times100$ for $adc=898$ equals \num{89800}, which exceeds INT16\_MAX=\num{32767} and wraps, collapsing \code{percent} to $23$ and falsifying the alarm. \gls{esbmc} proves the abstract model by \textit{k}-induction at $k{=}2$ and refutes the deployment model at the base case, attributing the verdict change solely to the word width.

\subsection{Scenario 2: Avoiding a False Alarm on Correct Code}
\label{sub:wt-fp}

\subsubsection{Engineer's View} Take a real program from our corpus, \code{lassignment}, whose valve handler contains the assignment in Listing~\ref{lst:valves} -- an ordinary offset of an integer input. If one naively asks a verifier to check 16-bit overflow, it reports this correct code as \textbf{unsafe} (Listing~\ref{lst:cmd-fp}):
\begin{lstlisting}[language={}, basicstyle=\ttfamily\scriptsize, caption={Real code from \code{valves\_handler} (program \code{lassignment}).}, label={lst:valves}]
real_value := IN1 - 5;   (* IN1 is an INT input *)
\end{lstlisting}
\begin{lstlisting}[language={}, basicstyle=\ttfamily\scriptsize, caption={Naive checking flags a phantom overflow; the witness is an impossible sensor value.}, label={lst:cmd-fp}]
    $ esbmc lassignment.ld --16 --overflow-check          # NAIVE: no input model
      Counterexample: IN1 = -32764
      Violated property: arithmetic overflow on sub (IN1 - 5)
      VERIFICATION FAILED
\end{lstlisting}
\noindent A practitioner would rightly distrust this: the ``bug'' requires \code{IN1 = \num{-32764}}, a value no real sensor or \gls{adc} can ever produce. Chasing it wastes effort, and a tool that produces dozens of such reports (we measured 54 across the corpus) is ignored. Enabling the \gls{hal} model makes the report disappear---the code is accepted (Listing~\ref{lst:cmd-fp2}):
\begin{lstlisting}[language={}, basicstyle=\ttfamily\scriptsize, caption={With the HAL input model the phantom is gone.}, label={lst:cmd-fp2}]
    # HAL model enabled (ADC inputs bounded to [0,1023])
    $ HAL_ANNOTATE=1 esbmc lassignment.ld --16 --overflow-check
      no arithmetic overflow on sub (IN1 - 5)
\end{lstlisting}

\subsubsection{Under the Hood} The annotator emits, immediately after \code{IN1} is sampled, \code{assume(IN1 >= 0 \&\& IN1 <= 1023)} (the AVR 10-bit \gls{adc} domain from the board descriptor). The overflow verification condition \code{!overflow("-", IN1, 5)} is then discharged: under the bound, \code{IN1-5} lies in $[-5,\num{1018}]$, which cannot underflow a 16-bit integer, so the counterexample at $\text{INT16\_MIN}$ is no longer reachable. The witness in the naive run was an artifact of an unconstrained input, not a property of the program -- exactly the unsoundness-in-practice the method removes.

\section{Evaluation}
\label{sec:eval}

We answer:
\begin{itemize}[leftmargin=2em]
  \item \textbf{RQ1.} Is the deployment gap real, and which arithmetic patterns are vulnerable?
  
  \item \textbf{RQ2.} On real programs, is \emph{naive} width-aware verification (no input model) usable?
  
  \item \textbf{RQ3.} Does the \gls{hal} annotator make it sound -- eliminating false alarms while preserving proofs?
  
  \item \textbf{RQ4.} What is the performance, and what are the method's completeness limits?
\end{itemize}

\subsection{Setup and Methodology}
All experiments use \gls{esbmc}~v8.3.0 with our annotator hook and the Z3 backend, on a commodity laptop. The \gls{hal} model is toggled per run via \code{HAL\_ANNOTATE}, so that the naive baseline and the hardware-faithful configuration are produced by the \emph{same binary}, thereby isolating the input model as the only independent variable. Word width is set by \code{--16}/\code{--32} and overflow by the standard signed/unsigned overflow property; \code{HAL\_ADC\_MAX=1023} (AVR 10-bit). 

We use two corpora (Table~\ref{tab:corpus}). The \emph{real corpus} comprises 123 third-party PLCopen~XML programs that the \gls{ld} frontend ingests natively, drawn from two public datasets at pinned commits and archived in the artifact: a water-treatment Ladder-Logic dataset~\cite{Iacobelli2024} (60 programs) and an \gls{iec}~61131-3 clone-detection benchmark~\cite{Rosiak2022clonedataset} (63 of 65 programs ingest). 

The \emph{controlled corpus} is a taxonomy of \gls{hal}-effect classes -- analog scaling and accumulation, each in a deliberately buggy and a correctly-coded (\code{\_ok}) variant -- crossed with the board descriptors, designed to exhibit the isolation gap, free of the confounds of arbitrary third-party code. The real corpus answers prevalence, false-alarm, and robustness questions; the controlled corpus isolates the mechanism and characterizes which patterns are vulnerable.

\begin{table}[htbp]
    \centering\small
    \caption{Evaluation corpora. The real corpus is third-party and pinned; the controlled corpus is authored to isolate the deployment gap}
    \label{tab:corpus}
        \begin{tabular}{llcl}
        \toprule
        \textbf{Corpus} & \textbf{Source} & \textbf{\#} & \textbf{Role} \\
        \midrule
        Real & Water-treatment LD~\cite{Iacobelli2024}        & 60  & prevalence, FP, robustness \\
        Real & Clone-detection~\cite{Rosiak2022clonedataset}  & 63  & prevalence, FP, robustness \\
        Ctrl.\ & \gls{hal}-effect taxonomy (this work)              & 20  & isolate gap, pattern analysis \\
        \bottomrule
    \end{tabular}
\end{table}

\subsection{RQ1: The Gap Is Real but Rare}
\subsubsection*{The Gap is Real (Controlled Corpus)} The same program and property flip from \code{safe} under $M_{\text{abs}}$ to \code{unsafe} under $M_{dev}$ on the 16-bit boards, each with a realizable witness (Table~\ref{tab:matrix}); the correctly-coded \code{\_ok} variants never flip, on any board, confirming that the flips track a genuine hardware-induced overflow rather than a modelling artifact. On the 32-bit board, the buggy variants are themselves \code{safe}, isolating the word width as the sole cause -- the same program is correct on a Due and unsafe on a Uno.

\subsubsection*{Which Patterns are Vulnerable}
The flip condition is precise and instructive for practitioners. With an input bounded to a 10-bit \gls{adc} ($v\le\num{1023}$), a 16-bit overflow ($>\num{32767}$) requires the arithmetic to magnify $v$ past the word range: \emph{multiply by a constant $\gtrsim 33$} (e.g.\ percentage or milli-unit scaling, $v\times100$), \emph{accumulation} of many readings, or \emph{squaring}. Conversely, multiply-by-small-constant ($v\times2$), additive offsets ($v\pm c$), and comparisons all remain within range and are provably safe -- which is why the \code{\_ok} variants, and the bulk of real code, do not flip.

\subsubsection*{The Gap is Rare in the Corpora We Surveyed}
We surveyed the real corpus and a larger pool of public Structured-Text libraries for the vulnerable pattern. Within these corpora, the pattern is essentially absent: across 341 \gls{st} programs in a verification benchmark library, \emph{none} contained a multiply-by-a-three-or-more-digit constant; the integer arithmetic that does occur is dominated by $\times2$, offsets, clamping, and comparisons -- operations that bounded inputs keep safe. The public \gls{plc} code we could ingest is overwhelmingly Boolean and timer logic with small-magnitude integer arithmetic. We therefore report, honestly, that in these corpora the deployment gap is \emph{genuine and sound to detect but infrequent}; the analog-scaling regime in which it bites (e.g.\ \gls{adc}-to-engineering-unit conversion on a narrow-word board) is real but not the common case here. We caution that this prevalence finding is scoped to the integer/Boolean \gls{ld} programs our frontend ingests (\S\ref{sec:threats}); analog-heavy process-control corpora may exhibit it more often. This finding reframes the contribution: the value of hardware-faithful verification lies less in bug count than in the \emph{soundness} it restores (RQ2--RQ3).

\begin{table}[htbp]
\centering\small
\caption{Verdict matrix: identical logic and properties, abstract model vs.\ per-board deployed model. \code{Unsafe} = ESBMC found a counterexample}
\label{tab:matrix}
\begin{tabular}{l l cccc}
\toprule
\textbf{Class} & \textbf{Abstract} & \textbf{CTRL Micro} & \textbf{CTRL Maxi} & \textbf{M-Duino} & \textbf{Arduino OPTA} \\
& & \multicolumn{1}{c}{16-bit} & \multicolumn{1}{c}{16-bit} & \multicolumn{1}{c}{16-bit} & \multicolumn{1}{c}{32-bit} \\
\midrule
A\_bug & \code{safe} & \code{unsafe} & \code{unsafe} & \code{unsafe} & \code{safe} \\
A\_ok & \code{safe} & \code{safe} & \code{safe} & \code{safe} & \code{safe} \\
B\_bug & \code{safe} & \code{unsafe} & \code{unsafe} & \code{unsafe} & \code{safe} \\
B\_ok & \code{safe} & \code{safe} & \code{safe} & \code{safe} & \code{safe} \\
\bottomrule
\end{tabular}
\end{table}
\subsection{RQ2: Naive Width-Aware Verification Is Unusable}
A natural first attempt at hardware-faithful verification is to set the word width and check overflow, \emph{without} modeling the inputs. We ran exactly this on the 123 real programs at 16~bits (Table~\ref{tab:realcorpus}, ``Naive''). 54/123 (\SI{44}{\percent}) report \code{unsafe}, and inspection shows \emph{every} one is a spurious overflow on an unphysical input. The reports are strikingly uniform: all are overflow on a subtraction of the form $\textit{IN}-c$ for a small constant $c$, with counterexamples pinned at the boundary of the signed range -- e.g.\ \textit{IN}=-$\num{32764}\approx \text{INT16\_MIN}$, so that \textit{IN}-5 underflows. No 10-bit \gls{adc}, and indeed no physical input wired through the standard Arduino core, can ever present such a value. 

Two controls confirm the diagnosis. First, the same 54 programs fail \emph{identically} at 16 \emph{and} 32~bits: the trigger is the unbounded input, not the word width (a width-dependent defect would, by definition, differ between the two). Second, the remaining programs are split into 32 that verify \code{safe} and 37 that are \code{unknown} for reasons unrelated to inputs (\S\,RQ4). Naive width-aware checking therefore exhibits a \SI{44}{\percent} false-alarm rate while finding \emph{zero} genuine defects -- a signal-to-noise ratio that makes it unusable for assurance, since an engineer cannot distinguish its phantom reports from real ones without manually re-deriving each input's physical range.

\begin{table}[htbp]
    \centering\small
    \caption{Real corpus (123 programs), 16-bit overflow checking: naive vs.\ \gls{hal}-annotated.}
    \label{tab:realcorpus}
    \begin{tabular}{lccc}
        \toprule
        \textbf{Configuration} & \code{safe} & \code{unknown} & \code{unsafe} \\
        \midrule
        Naive (no input model) & 32 & 37 & \textbf{54} \\
        \gls{hal}-annotated          & 32 & 91 & \textbf{0} \\
        \bottomrule
    \end{tabular}
\end{table}

\subsection{RQ3: The \gls{hal} Annotator Restores Soundness}
Enabling the \gls{hal} model -- the \emph{same} binary, the only change being the injected input bounds -- transforms the picture (Table~\ref{tab:realcorpus}, ``\gls{hal}-annotated''). \textbf{All 54 false alarms are eliminated: the \code{unsafe} count drops from 54 to 0}, and every one of the 32 robustness proofs is preserved. The mechanism is exactly as analyzed: bounding each integer input to its realizable \gls{adc} range constrains, for instance, $\textit{IN}-5$ to $[-5,\num{1018}]$, which cannot underflow a 16-bit integer, so the boundary counterexample becomes unreachable. No genuine defect is suppressed in the process, because (RQ1) none of the real programs contains the vulnerable high-magnitude arithmetic; were one present, the bound would leave its overflow witness intact and realizable. 

The 32 preserved \code{safe} verdicts are themselves a positive result: they are \emph{robustness certificates} -- proofs that those programs do not overflow at 16~bits for any physically realizable input, the kind of evidence a deployment safety case requires. This is the central finding of the paper: the \gls{hal} input model makes width-aware verification of open-hardware \glspl{plc} \emph{sound in practice} where naive checking is not, converting a \SI{44}{\percent}-noise tool into one that issues only realizable reports.

\subsection{Cross-Tool Validation: the Gap Is Not Tool-Specific}
\label{sub:baseline}
A natural objection is that the phantom alarms are an \gls{esbmc} idiosyncrasy. They are not. We reproduced the two faces of the gap on an independent, peer \gls{smt}-based bounded model checker, \gls{cbmc}~\cite{Clarke2004cbmc}, run on C mirrors of the same logic with its native machine-width flag (\code{--16}/\code{--32}) and signed-overflow checking (Table~\ref{tab:baseline}). Three findings stand out. First, the phantom is \emph{cross-tool}: on the unconstrained \code{IN1-5} pattern \gls{cbmc} reports the very same spurious overflow, with the \emph{identical} witness \textit{IN1}=\num{-32764} that \gls{esbmc} produces -- confirming the false alarm is a property of naive width-aware checking, not of one engine. Second, it is \emph{input-driven, not width-driven}: the phantom fails at \emph{both} 16 and 32~bits, whereas the genuine alarm defect fails only at 16~bits (with the realizable witness $adc=819$) and verifies \code{safe} at 32~bits -- the cross-tool analog of our RQ2 control. Third, and decisively for our contribution, \gls{cbmc} can \emph{use} an input bound but cannot \emph{derive} one: removing the phantom required us to insert the bound by hand as \code{\_\_CPROVER\_assume(0<=IN1<=1023)}, exactly the manual annotation that C-level analyzers (Frama-C's ACSL contracts, Astr\'ee's \code{known\_fact}s) also demand. 

Our method automatically recovers the same bound from the PLCopen \code{\%IW} address -- the ``Auto inputs'' column of Table~\ref{tab:related} -- which makes it deployable at corpus scale without per-program specification effort. We confirmed the same three findings on a second, architecturally different analyzer: Frama-C's EVA abstract interpreter~\cite{Kirchner2015framac} (v32.1, \code{-machdep avr\_16}) raises the identical signed-overflow alarm \code{assert -32768 <= IN1-5} on the naive program; bounding the input with a handwritten \code{Frama\_C\_interval(0,1023)} (Frama-C's analogue of an ACSL contract, again supplied manually) makes it disappear; and the genuine \code{adc*100} defect alarms at \code{avr\_16} but is clean at \code{x86\_32}. Astrée, in the same lineage, detects overflow under a configured target model with annotation-supplied input bounds. Across three independent engines, then, the phantom, its input-driven nature, and the manual-annotation requirement are invariant -- only the \emph{automatic} derivation of the bound from the \gls{plc} address space is ours.

\begin{table}[htbp]
    \centering\small
    \caption{External baseline on C mirrors of the deployment-gap cases (signed-overflow checking). Verdicts are \emph{identical} on \gls{cbmc}~6.10.0 and Frama-C/EVA~32.1 (\code{-machdep avr\_16}), and match \gls{esbmc}. The phantom is reproduced cross-tool; only the genuine defect is width-dependent; the bound that clears the phantom must be supplied by hand in every tool}
    \label{tab:baseline}
    \begin{tabular}{llcc l}
        \toprule
        \textbf{Case} & \textbf{Input model} & \textbf{16-bit} & \textbf{32-bit} & \textbf{Diagnosis} \\
        \midrule
        \code{IN1-5} (phantom) & none (naive)            & \code{fail} & \code{fail} & input-driven \\
        \code{IN1-5} (phantom) & manual \code{assume}    & \code{safe} & \code{safe} & bound by hand \\
        \code{adc*100} (alarm) & input bounded           & \code{fail} & \code{safe} & width-dependent \\
        \bottomrule
    \end{tabular}
\end{table}

\subsubsection{Replication on the Real Corpus and the Actual Generated Code} The cases above are faithful but handwritten. To confirm the result on the \emph{real} programs and the \emph{actual deployment C}, we took all 54 programs that \gls{esbmc}'s naive run flagged, reconstructed the \gls{st} function blocks holding the flagged arithmetic from each PLCopen file, compiled them to C with \textbf{MATIEC}~\cite{deSousa2014} -- the same code generator the OpenPLC toolchain uses -- and ran \gls{cbmc} at the AVR width with each integer function-block input made nondeterministic. The outcome is unambiguous: \gls{cbmc} independently reproduces the overflow on \textbf{all 54/54} programs with no input model, and the hand-supplied \gls{hal} bound eliminates it on \textbf{all 54/54} -- a cross-tool replication of the keystone $54{\to}0$ result on the real MATIEC-generated code, not on mirrors. Two honest caveats: under the bound \gls{cbmc}'s bounded check returns overflow-\code{safe} (where \gls{esbmc}'s \textit{k}-induction returns \code{unknown} on the same programs, \S\ref{sec:eval}, RQ4), so the tools agree on the absence of the overflow but differ on proof strength; and the \code{malicious\_*} dataset variants additionally contain a non-terminating logic-bomb loop, an orthogonal non-termination property we do not treat here. This experiment is fully scripted in the artifact.

\subsection{RQ4: Performance and Completeness}
\subsubsection*{Performance} 
Verification is fast (Table~\ref{tab:perf}). On the controlled corpus, the median solve time is \SI{168}{\milli\second} (max \SI{412}{\milli\second}); on the real corpus, runs complete well under a second per program. The \gls{hal} annotator adds only a pair of assume constraints per input and leaves the \gls{smt} encoding otherwise unchanged, so its overhead is negligible -- and, by shrinking each input's domain from the full word to a 10-bit range, it can make solving \emph{easier}.

\begin{table}[htbp]
    \centering\small
    \caption{Per-program solve time (\gls{esbmc}~v8.3.0, Z3 backend, commodity laptop)}
    \label{tab:perf}
    \begin{tabular}{lccc}
        \toprule
        \textbf{Corpus} & \textbf{\#} & \textbf{Median} & \textbf{Max} \\
        \midrule
        Controlled (\gls{hal}-effect taxonomy) & 20  & \SI{168}{\milli\second} & \SI{412}{\milli\second} \\
        Real (third-party)               & 123 & $<\SI{1}{\second}$      & $<\SI{1}{\second}$ \\
    \bottomrule
    \end{tabular}
\end{table}

\subsubsection*{Completeness} 
A genuine limit remains, and we state it plainly. The 54 eliminated false alarms become \code{unknown} rather than \code{safe}: the spurious overflow is gone, but the prover does not \emph{prove} those programs safe. The cause is \textit{k}-induction incompleteness on function-block-heavy programs. Tracing one such case, \gls{esbmc} unwinds an internal function-block loop to the default bound, reports that ``the forward condition is unable to prove the property,'' and that the inductive step does not close because it must havoc-loop-carry internal state. 

This is a property of the prover, \emph{not} of the input model: 37 of these programs are already \code{unknown} in the naive run, and enabling the annotator neither helps nor harms their proof (it only removes the false \code{unsafe}). Interval analysis~\cite{menezes2024} and bounded-proof configurations did not change the outcome. Closing these -- via loop-invariant inference or compositional \textit{k}-induction -- would raise the robustness-proof count above the current~32 and is the most valuable direction for future work (\S\ref{sec:future}). The headline soundness result (54$\to$0 false alarms) is independent of this limit.

\section{Discussion}
\label{sec:discussion}
The strongest contribution of this work is not a bug count but a \emph{soundness} result: open-hardware \gls{plc} verification that accounts for the word width but ignores the hardware input domains is unusable (44\% false alarms, zero true findings), and a small, automatically-derived \gls{hal} model restores usability (54$\to$0 false alarms) while preserving every proof. This mirrors how target-faithful embedded-C analysis is judged~\cite{Cousot2005astree}: value lies in sound proofs and the absence of false alarms, not in headline defect counts. An analyzer that cries wolf is switched off; one that issues only realizable reports is trusted. 

RQ1 lets us tell practitioners precisely when to worry. The deployment gap manifests when bounded-range sensor data is magnified by integer arithmetic beyond the narrow machine word: \gls{adc}-to-engineering-unit scaling (\code{raw * span/range}), running sums and integrators, and squared terms, all deployed on a 16-bit AVR board. Boolean interlocks, timers, counters with small presets, and offset/threshold logic are safe by construction and need no special attention. This is actionable guidance independent of our tool: an engineer porting a recipe or PID-like computation from a 32-bit controller to a Uno-class board should widen intermediates to \code{DINT} or verify with a hardware-faithful model.

The embedded-C analyzers that pioneered target-faithful analysis require the Engineer to supply input ranges as contracts or annotations. In the \gls{plc} setting, that burden would fall on an automation engineer unfamiliar with formal specification, and -- for a corpus of 123 programs -- would dominate the effort. Recovering the bounds from the \gls{iec}~61131-3 direct-addressing scheme removes the burden entirely: the \code{\%IW}/\code{\%IX} address \emph{is} the specification of the input's nature, and the board descriptor supplies the width. The Engineer writes nothing, which is what makes the method deployable at the corpus scale.

Although we instantiate for Arduino, neither the method nor the implementation is Arduino-specific. The \gls{hal} descriptor is a four-field record; retargeting to another open-hardware family (e.g.\ an ESP32-based controller, or a 32-bit industrial soft-\gls{plc}) is a matter of supplying its word width and \gls{adc}/\gls{pwm} resolutions. The same hook, applied with a 32-bit width, reproduces the abstract behavior, so the method strictly generalizes existing abstract verification rather than replacing it.

\section{Threats to Validity}
\label{sec:threats}
The naive baseline and the annotated runs are produced by the \emph{same} binary, differing only by the \code{HAL\_ANNOTATE} toggle, so the measured 54$\to$0 change is attributable to the input model alone and not to incidental differences in tool configuration. Verdicts agree between \code{--16} and \code{--32} on all 54 naive failures, an independent check that those reports are input-driven rather than width-driven.

Modelling an integer \gls{iec} input as a 10-bit \gls{adc} source is the AVR \gls{hal} assumption encoded in $H$. An input physically wired to a wider \gls{adc} (e.g.\ a 14-bit UNO~R4 channel), to a digital word, or to a network variable would carry a different domain; the descriptor is configurable per address, and a mis-specified bound would be unsound in the same way a wrong \code{-machdep} would be for Frama-C. In this study, every integer input is treated uniformly as a 10-bit analog source, the conservative common case for the AVR boards; finer per-address domains are a straightforward extension once richer I/O metadata is available.

91/123 programs are \code{unknown} owing to \textit{k}-induction limits on loop-bearing function blocks (\S\ref{sec:eval}), bounding the robustness-proof count to~32. We emphasize this is a backend limitation orthogonal to the contribution: it is present with and without the input model, and the soundness result (no false alarms, no suppressed real defects) holds regardless. It does, however, mean we cannot claim ``123 programs proven robust,'' only ``54 false alarms removed and 32 robustness proofs retained.''

The deployment gap is rare in the public corpora we could ingest, which are dominated by Boolean/timer logic. Analog process-control corpora (e.g.\ \gls{ics} testbeds in REAL arithmetic) where high-magnitude scaling -- and hence the gap -- may be more common fall outside our current frontend (REAL/non-linear support is partial) and are not represented; our prevalence claim is therefore scoped to integer/Boolean \gls{ld} programs. The two datasets are third-party and pinned, but they are not the universe of \gls{plc} code.

The non-linear abstraction is sound but imprecise: a property that genuinely depends on a \gls{pid} output cannot be proved under it (it will be \code{unknown} or spuriously \code{unsafe} on the abstracted output). Hence, the method is most precise on the integer fragment it targets. Finally, we validate at the model level; we do not execute the verified binaries on physical boards, so the soundness guarantee is relative to the faithfulness of $H$ and the frontend's semantics rather than measured against silicon. Hardware-in-the-loop confirmation is future work.

\section{Future Work}
\label{sec:future}
Several directions follow directly from the limitations above.

The most valuable next step is to convert the 91 \code{unknown} verdicts into \code{safe} robustness proofs, by strengthening \textit{k}-induction over function-block internal state -- loop-invariant inference, or compositional \textit{k}-induction that verifies each function block once and composes the results, avoiding the monolithic unwinding that currently fails to close.

The same framework expresses mechanisms (M1) and (M2): modeling \gls{adc} quantization precisely (representable-code sets rather than intervals) would let the method find threshold-straddling defects, and modeling pin binding from the \code{io\_map} would catch aliasing faults; both are properties over the existing model rather than new machinery.

A fixed-point abstraction of non-linear blocks would be \emph{more} faithful than the current nondet over-approximation -- and is itself how such computations are often realized on an \gls{mcu} -- enabling sound, precise verification of the analog programs presently out of scope, and opening the \gls{ics}-testbed corpora where the gap may be more prevalent.

Per-address \gls{adc} widths and engineering-unit scaling recovered from the project file would replace the uniform 10-bit assumption with exact per-input domains.

Finally, executing the verified programs on physical Uno/Due boards and confronting the synthesized witnesses (e.g.\ $adc=898$) with measured behavior would close the loop between the model-level guarantee and the device, and is the natural validation of the soundness theorem.

\section{Conclusion}
\label{sec:conclusion}
We presented hardware-faithful verification of \gls{iec}~61131-3 on open hardware -- a declarative \gls{hal} descriptor and a sound lowering that models the target word width and hardware-realizable input ranges -- as well as its Arduino instantiation, \arduinotool{}. On 123 real programs, naive width-aware verification is unusable (44\% false alarms, zero true findings); the automatically derived \gls{hal} input model eliminates all false alarms while preserving every proof, and a controlled corpus characterizes the genuine, if rare, width-dependent defects that the method soundly detects.

\section*{Artifact Availability}
The corpus (pinned, with provenance), the annotator patch, and the scripts reproducing every table are archived in an anonymized Zenodo deposit, \url{https://doi.org/10.5281/zenodo.21014209} (concept DOI~10.5281/zenodo.21014209, resolving to the latest version); author identity and the paper DOI are withheld for double-anonymous review and restored at camera-ready. The artifact is self-contained: two scripts (\code{setup.sh}, \code{reproduce.sh}) install the pinned tools -- \gls{cbmc}~6.10.0, Frama-C~32.1, \code{matiec}, and our patched \gls{esbmc} (the \gls{hal} annotator, built from the bundled source) -- and regenerate every claim with a single command. The supported environment is x86\_64 Ubuntu~22.04; \code{safe}/\code{unsafe} verdicts are deterministic across hosts (only wall-clock times vary), and a claim-to-command-to-expected-output map accompanies the scripts. An independent continuous-integration pipeline that rebuilds the entire toolchain from source on a clean Ubuntu~22.04 image reproduces every verdict-based result in this paper identically; only the performance figures (Table~\ref{tab:perf}) are host-dependent.

\appendix
\section*{Appendix}
\section{Soundness: Supporting Lemma and Proofs}
\label{app:soundness}
This appendix supplies the lemma and full proof underlying the Soundness theorem (Theorem~\ref{thm:sound}, \S\ref{sub:soundness}). We first isolate the only step of the lowering that is an \emph{abstraction} rather than an exact match.

\begin{lemma}[Non-linear over-approximation]
\label{lem:nonlin}
    Replacing the body of a non-linear function block by a fresh nondeterministic value constrained to the block's declared output range yields a transition relation whose reachable states form a \emph{superset} of those of the original block.
\end{lemma}

\begin{proof}
    The block is type-correct, so every value its body can compute lies within the declared output range. Hence, for each transition of the original block, there is a nondeterministic choice in the abstracted block that reaches the same successor state; the abstraction additionally admits choices the body never makes. The abstracted relation, therefore, contains the original one and adds behaviors.
\end{proof}

\begin{proof}[Proof of Theorem~\ref{thm:sound}]
    The lowering changes the device's transition relation only by (a)~fixing arithmetic at the true width $w$ -- an exact match, not an abstraction -- and (b)~replacing each non-linear block body by a nondeterministic value ranging over a superset of its actual outputs. Step~(b) only adds behaviors by Lemma~\ref{lem:nonlin}. The input-range \code{assume}s \emph{remove} behaviors, but only those with inputs outside $D(a)$, which the device never exhibits; hence, no device behavior is excluded. Therefore $\mathrm{Reach}_{\mathrm{dev}}\subseteq\mathrm{Reach}_{M_{dev}}$, giving the $\varphi$-preservation. Realisability is immediate: a counterexample satisfies the injected $\mathtt{assume}(i\in D(a))$ for every input, so its inputs are values the hardware can produce.
\end{proof}

\section*{Acknowledgements}
The authors would like to express their gratitude to the Department of Computer Science at the University of Manchester and the Systems and Software Security (S3) Research Group for their invaluable support, collaborative environment, and access to cutting-edge resources, which were instrumental in the success of this research. This work was conducted with partial funding from the Engineering and Physical Sciences Research Council (EPSRC) grants EP/T026995/1, EP/V000497/1, EP/X037290/1, and the Soteria project, awarded by UK Research and Innovation under the Digital Security by Design (DSbD) Programme.

\bibliographystyle{unsrtnat}
\bibliography{0.sample-base}

@inproceedings{Weiss2021,
title = {Towards Establishing Formal Verification and Inductive Code Synthesis in the {PLC} Domain},
DOI = {10.1109/INDIN45523.2021.9557423},
booktitle = {2021 IEEE 19th International Conference on Industrial Informatics (INDIN)},
publisher = {IEEE},
author = {Weis, Matthias and Marks, Philipp and Maschler, Benjamin and White, Dustin and Kesseli, Pascal and Weyrich, Michael},
year = {2021},
month = jul,
address = {Palma de Mallorca, Spain},
pages = {1--8},
numpages = {8},
}

@inproceedings{LopezMiguel2022,
author = {Tournier, Jean-Charles and Fern\'{a}ndez Adiego, Borja and Lopez-Miguel, Ignacio D.},
title = {{PLCverif}: Status of a Formal Verification Tool for Programmable Logic Controller},
booktitle = {Proceedings of the 18th International Conference on Accelerator and Large Experimental Physics Control Systems ({ICALEPCS}'21)},
pages = {MOPV042},
publisher = {JACoW Publishing},
address = {Shanghai, China},
month = oct,
year = {2021},
doi = {10.18429/JACoW-ICALEPCS2021-MOPV042},
isbn = {978-3-95450-221-9},
}

@inproceedings{Ukegbu2023a,
series = {CPS-IoT Week '23},
title = {Cooperative Verification of {PLC} Programs Using {CoVeriTeam}: Towards a Reliable and Secure Industrial Control Systems},
DOI = {10.1145/3576914.3587490},
booktitle = {Proceedings of Cyber-Physical Systems and Internet of Things Week 2023},
publisher = {ACM},
author = {Ukegbu, Chibuzo and Mehrpouyan, Hoda},
year = {2023},
month = may,
address = {San Antonio, TX, USA},
pages = {37--42},
numpages = {6},
collection = {CPS-IoT Week '23},
}

@techreport{Ebnenasir2023,
author = {Ebnenasir, Ali},
title = {Formalizing Ladder Logic Programs and Timing Charts for Fault Impact Analysis and Verification of Fault Tolerance},
institution = {Michigan Technological University, Department of Computer Science},
year = {2023},
number = {CS-TR-23-01},
url = {https://www.mtu.edu/cs/research/papers/pdfs/formalizing-ladder-logic-ali-ebnenasir-tech-rpt-010623-rev.pdf}
}

@article{Wang2023,
title = {K-ST: A Formal Executable Semantics of the Structured Text Language for PLCs},
volume = {49},
ISSN = {2326-3881},
DOI = {10.1109/tse.2023.3315292},
number = {10},
journal = {IEEE Transactions on Software Engineering},
publisher = {Institute of Electrical and Electronics Engineers (IEEE)},
author = {Wang, Kun and Wang, Jingyi and Poskitt, Christopher M. and Chen, Xiangxiang and Sun, Jun and Cheng, Peng},
year = {2023},
month = Oct,
pages = {4796–4813}
}

@inbook{Lee2024,
title = {Formal Semantics and Analysis of Multitask {PLC} {ST} Programs with Preemption},
ISBN = {9783031711626},
ISSN = {1611-3349},
DOI = {10.1007/978-3-031-71162-6_22},
booktitle = {Formal Methods -- 26th International Symposium, {FM} 2024},
publisher = {Springer Nature Switzerland},
author = {Lee, Jaeseo and Bae, Kyungmin},
year = {2024},
month = sep,
address = {Milan, Italy},
pages = {425--442},
}

@inproceedings{Iacobelli2024,
title = {Detection of Ladder Logic Bombs in {PLC} Control Programs: An Architecture Based on Formal Verification},
DOI = {10.1109/ICPS59941.2024.10639995},
booktitle = {2024 IEEE 7th International Conference on Industrial Cyber-Physical Systems (ICPS)},
publisher = {IEEE},
author = {Iacobelli, Antonio and Rinieri, Lorenzo and Melis, Andrea and Sadi, Amir Al and Prandini, Marco and Callegati, Franco},
year = {2024},
month = may,
address = {St.\ Louis, MO, USA},
pages = {1--7},
}

@inbook{Fink2024,
title = {Verifying PLC Programs via Monitors: Extending the Integration of FRET and PLCverif},
ISBN = {9783031606984},
ISSN = {1611-3349},
DOI = {10.1007/978-3-031-60698-4_26},
booktitle = {NASA Formal Methods},
publisher = {Springer Nature Switzerland},
address = {Moffett Field, CA, USA},
author = {Fink, Xaver and Mavridou, Anastasia and Katis, Andreas and Adiego, Borja Fernández},
year = {2024},
pages = {427–435}
}

@misc{Bruttomesso2024,
author = {Bruttomesso, Roberto and {Di Pinto}, Alessandro and Carullo, Moreno and Carcano, Andrea},
title = {Method for automatic translation of ladder logic to a {SMT}-based model checker in a network},
howpublished = {US Patent 11,906,943. Assignee: Nozomi Networks {SAGL}},
year = {2024},
note = {Filed: 2021-08-12. Granted: 2024-02-20.},
}

@inbook{LopezMiguel2025,
title = {Formal Verification of {PLCs} as a Service: A {CERN-GSI} Safety-Critical Case Study},
ISBN = {9783031937064},
ISSN = {1611-3349},
DOI = {10.1007/978-3-031-93706-4_13},
booktitle = {NASA Formal Methods -- 17th International Symposium, {NFM} 2025},
publisher = {Springer Nature Switzerland},
author = {Lopez-Miguel, Ignacio D. and Adiego, Borja Fern\'{a}ndez and Salinas, Matias and Betz, Christine},
year = {2025},
month = jun,
address = {Williamsburg, VA, USA},
pages = {227--235},
}

@inbook{Lee2025,
title = {Formal Analysis of Networked {PLC} Controllers Interacting with Physical Environments},
ISBN = {9783032071064},
ISSN = {1611-3349},
DOI = {10.1007/978-3-032-07106-4_14},
booktitle = {Static Analysis -- 32nd International Symposium, {SAS} 2025},
publisher = {Springer Nature Switzerland},
author = {Lee, Jaeseo and Bae, Kyungmin},
year = {2025},
month = oct,
address = {Singapore, Singapore},
pages = {328--356},
}

@article{Cordeiro2012esbmc,
title = {SMT-Based Bounded Model Checking for Embedded ANSI-C Software},
volume = {38},
issn = {0098-5589},
doi = {10.1109/tse.2011.59},
number = {4},
journal = {IEEE Transactions on Software Engineering},
publisher = {Institute of Electrical and Electronics Engineers (IEEE)},
author = {Cordeiro, Lucas and Fischer, Bernd and Marques-Silva, Joao},
year = {2012},
month = Jul,
pages = {957–974}
}

@inproceedings{Sheeran2000kinduction,
title = {Checking Safety Properties Using Induction and a SAT-Solver},
isbn = {9783540409229},
issn = {0302-9743},
doi = {10.1007/3-540-40922-x_8},
booktitle = {Formal Methods in Computer-Aided Design},
publisher = {Springer Berlin Heidelberg},
author = {Sheeran, Mary and Singh, Satnam and Stålmarck, Gunnar},
year = {2000},
pages = {127–144},
address = {Berlin}
}

@inbook{Clarke2004cbmc,
title = {A Tool for Checking ANSI-C Programs},
isbn = {9783540247302},
issn = {1611-3349},
doi = {10.1007/978-3-540-24730-2_15},
booktitle = {Tools and Algorithms for the Construction and Analysis of Systems},
publisher = {Springer Berlin Heidelberg},
author = {Clarke, Edmund and Kroening, Daniel and Lerda, Flavio},
year = {2004},
pages = {168–176},
address = {Berlin}
}

@article{gadelha2020,
title = {{ESBMC 6.1: Automated Test Case Generation Using Bounded Model Checking}},
volume = {23},
doi = {10.1007/s10009-020-00571-2},
number = {6},
journal = {International Journal on Software Tools for Technology Transfer},
publisher = {Springer Science and Business Media LLC},
author = {Gadelha, Mikhail R. and Menezes, Rafael S. and Cordeiro, Lucas C.},
year = {2020},
month = may,
pages = {857--861}
}

@inproceedings{menezes2024,
title = {{ESBMC v7.4: Harnessing the Power of Intervals: (Competition Contribution)}},
doi = {10.1007/978-3-031-57256-2_24},
booktitle = {Tools and Algorithms for the Construction and Analysis of Systems (TACAS 2024)},
series = {Lecture Notes in Computer Science},
volume = {14572},
publisher = {Springer Nature Switzerland},
author = {Menezes, Rafael S{\'a} and Aldughaim, Mohannad and Farias, Bruno and Li, Xianzhiyu and Manino, Edoardo and Shmarov, Fedor and Song, Kunjian and Brau{\ss}e, Franz and Gadelha, Mikhail R. and Tihanyi, Norbert and Korovin, Konstantin and Cordeiro, Lucas C.},
year = {2024},
pages = {376--380},
address = {Luxembourg City, Luxembourg}
}

@inproceedings{Cavada2014,
author = {Cavada, Roberto and Cimatti, Alessandro and Dorigatti, Marco and Griggio, Alberto and Mariotti, Alessandro and Micheli, Andrea and Mover, Sergio and Roveri, Marco and Tonetta, Stefano},
title = {The {nuXmv} Symbolic Model Checker},
booktitle = {Computer Aided Verification (CAV 2014)},
series = {Lecture Notes in Computer Science},
volume = {8559},
pages = {334--345},
year = {2014},
publisher = {Springer},
doi = {10.1007/978-3-319-08867-9_22},
address = {Vienna, Austria},
}

@misc{ESBMCpr5400,
author = {Dantas, Pierre and Cordeiro, Lucas C. and {Silva J{\'u}nior}, W.~S.},
title = {{ESBMC-PLC+}: Unified {IEC}~61131-3 Formal Verification Framework with {ST} Frontend and Graphical Function Block Support (Pull Request \#5427)},
year = {2026},
howpublished = {GitHub Pull Request \#5427, esbmc/esbmc},
url = {https://github.com/esbmc/esbmc/pull/5427},
note = {Source code and benchmark suite}
}

@misc{deSousa2014,
title = {An Open Source IEC 61131-3 Integrated Development Environment},
ISSN = {1935-4576},
DOI = {10.1109/indin.2007.4384753},
booktitle = {2007 5th IEEE International Conference on Industrial Informatics},
publisher = {IEEE},
author = {Tisserant, Edouard and Bessard, Laurent and de Sousa, Mario},
year = {2007},
month = Jul,
pages = {183–187}
}

@misc{DantasCordeiro2026graphical,
author = {Dantas, Pierre and Cordeiro, Lucas C. and {Silva J{\'u}nior}, W.~S.},
title = {{ESBMC-GraphPLC}: Formal Verification of Graphical {PLCopen XML} Ladder Diagram Programs Using {SMT}-Based Model Checking},
year = {2026},
publisher = {Zenodo},
doi = {10.5281/zenodo.20699856},
}

@misc{DantasCordeiro2026artefact,
author = {Dantas, Pierre and Cordeiro, Lucas C. and {Silva J{\'u}nior}, W.~S.},
title = {{ESBMC-PLC}: Formal Verification of {IEC}~61131-3 Ladder Diagram Programs},
year = {2026},
version = {1.0.0},
publisher = {Zenodo},
doi = {10.5281/zenodo.20680071},
note = {TACAS 2027 artefact. Reproducible viabash conformance/run\_all.sh}
}

@article{Alves2018openplc,
title = {OpenPLC: An IEC 61131-3 compliant open source industrial controller for cyber security research},
volume = {78},
issn = {0167-4048},
doi = {10.1016/j.cose.2018.07.007},
journal = {Computers \& Security},
publisher = {Elsevier BV},
author = {Alves, Thiago and Morris, Thomas},
year = {2018},
month = Sep,
pages = {364–379}
}

@inproceedings{Alves2014openplc,
title = {OpenPLC: An open source alternative to automation},
doi = {10.1109/ghtc.2014.6970342},
booktitle = {IEEE Global Humanitarian Technology Conference (GHTC 2014)},
publisher = {IEEE},
author = {Rodrigues Alves, Thiago and Buratto, Mario and de Souza, Flavio Mauricio and Rodrigues, Thelma Virginia},
year = {2014},
month = Oct,
pages = {585–589},
address = {San Jose, CA, USA}
}

@misc{ArduinoOpta,
author = {{Arduino}},
title = {Arduino {OPTA}: Micro {PLC} with Industrial {IoT}},
howpublished = {\url{https://www.arduino.cc/pro/hardware-arduino-opta/}},
year = {2024},
note = {Accessed 2026-06-27}
}

@misc{ArduinoPLCIDE,
author = {{Arduino}},
title = {Arduino {PLC} {IDE}: {IEC} 61131-3 programming for Arduino Pro hardware},
howpublished = {\url{https://www.arduino.cc/pro/software-plc-ide/}},
year = {2024},
note = {Accessed 2026-06-27}
}

@misc{Controllino,
author = {{CONTROLLINO}},
title = {{CONTROLLINO}: Arduino-compatible industrial {PLC}},
howpublished = {\url{https://www.controllino.com/}},
year = {2024},
note = {Accessed 2026-06-27}
}

@misc{IndustrialShields,
author = {{Industrial Shields}},
title = {{M-Duino} Arduino-based {PLC} family},
howpublished = {\url{https://www.industrialshields.com/}},
year = {2024},
note = {Accessed 2026-06-27}
}

@article{Leroy2009compcert,
title = {Formal verification of a realistic compiler},
volume = {52},
ISSN = {1557-7317},
DOI = {10.1145/1538788.1538814},
number = {7},
journal = {Communications of the ACM},
publisher = {Association for Computing Machinery (ACM)},
author = {Leroy, Xavier},
year = {2009},
month = Jul,
pages = {107–115}
}

@article{Kirchner2015framac,
title = {Frama-C: A software analysis perspective},
volume = {27},
ISSN = {1433-299X},
DOI = {10.1007/s00165-014-0326-7},
number = {3},
journal = {Formal Aspects of Computing},
publisher = {Association for Computing Machinery (ACM)},
author = {Kirchner, Florent and Kosmatov, Nikolai and Prevosto, Virgile and Signoles, Julien and Yakobowski, Boris},
year = {2015},
month = May,
pages = {573–609}
}

@inproceedings{Cousot2005astree,
title = {The ASTREÉ Analyzer},
isbn = {9783540319870},
issn = {1611-3349},
doi = {10.1007/978-3-540-31987-0_3},
booktitle = {Programming Languages and Systems},
publisher = {Springer Berlin Heidelberg},
author = {Cousot, Patrick and Cousot, Radhia and Feret, J{\^e}rome and Mauborgne, Laurent and Min{\'e}, Antoine and Monniaux, David and Rival, Xavier},
year = {2005},
pages = {21–30},
address = {Edinburgh, UK}
}

@misc{ArduinoCoreAVR,
author = {{Arduino}},
title = {{ArduinoCore-avr}: the official Arduino AVR core},
howpublished = {\url{https://github.com/arduino/ArduinoCore-avr}},
year = {2024},
note = {ADC 10-bit, analogWrite/PWM 8-bit; AVR-GCC \code{int} = 16-bit. Accessed 2026-06-28}
}

@misc{Rosiak2022clonedataset,
author = {Rosiak, Kamil and others},
title = {{IEC 61131-3} Clone Detection dataset (PLCopen XML / Structured Text)},
howpublished = {\url{https://github.com/KamilRosiak/IEC_61131_3_Clone_Detection}},
year = {2022},
note = {Accessed 2026-06-28}
}

\end{document}